\documentclass[a4paper,onecolumn,accepted=2024-11-26]{quantumarticle}
\pdfoutput=1
\usepackage[numbers, sort]{natbib}
\usepackage{comment}
\usepackage{hyperref}
\usepackage[utf8]{inputenc}

\usepackage{amsmath,amsfonts,amssymb, amsthm}
\usepackage[parfill]{parskip}
\usepackage[shortlabels]{enumitem}
\usepackage{amssymb}
\usepackage{physics}
\usepackage{braket}
\usepackage{xfrac}
\usepackage[singlelinecheck = false]{caption}

\usepackage{float}
\usepackage{graphicx}
\usepackage[justification=centering]{caption}
\usepackage{subcaption}
\usepackage{soul}
\usepackage{cancel} 
\usepackage{xcolor}
\usepackage[overload]{empheq} 
\usepackage[ruled,vlined]{algorithm2e}
\usepackage{algorithmic}
\usepackage{mathtools}
\DeclarePairedDelimiterX{\infdivx}[2]{(}{)}{%
  #1\;\delimsize\|\;#2%
}

\DeclareMathOperator*{\argmax}{arg\,max}

\def\bal{\begin{aligned}}
\def\eal{\end{aligned}}

\def\be{\begin{equation}}
\def\ee{\end{equation}}

\newcommand*\innerprod[2]{\left\langle #1, #2 \right\rangle}

\usepackage{scalerel}

\newtheorem{definition}{Definition}[section]
\newtheorem*{definition*}{Definition}
\newtheorem{lemma}{Lemma}[section]
\newtheorem*{lemma*}{Lemma}
\newtheorem{theorem}{Theorem}[section]
\newtheorem*{theorem*}{Theorem}
\newtheorem{corollary}{Corollary}[section]
\newtheorem*{corollary*}{Corollary}
\theoremstyle{plain}

\usepackage{graphicx}
\usepackage{caption}

\newtheorem{remark}{Remark}[section]

\usepackage[T1]{fontenc}

\newcommand{\R}[0]{\mathbb{R}}

\newcommand{\A}[0]{\mathcal{A}}
\newcommand{\B}[0]{\mathcal{B}}

\renewcommand{\H}[0]{\mathcal{H}}

\newlist{todolist}{itemize}{2}
\setlist[todolist]{label=$\square$}
\usepackage{pifont}
\usepackage{tikz}

\newcommand*\lmax[0]{\mathrm{\lambda_{\max}}}

\newcommand*\da{\mathrm{D}(\mathcal{A})}
\newcommand*\db{\mathrm{D}(\mathcal{B})}

\newcommand*\id{\mathbb{I}}
\usepackage{bbm}
\newcommand{\edits}[1]{{\color{black}{#1}}}

\title{No-Regret Learning and Equilibrium Computation in Quantum Games
}
\author{Wayne Lin}
\email{wayne\_lin@sutd.edu.sg}
\author{Georgios Piliouras}
\email{georgios@sutd.edu.sg}
\author{Ryann Sim}
\email{ryann\_sim@sutd.edu.sg}
\author{Antonios Varvitsiotis}
\email{antonios@sutd.edu.sg}
\affiliation{Singapore University of Technology and Design, Singapore}

\date{}
\begin{document}

\maketitle

\begin{abstract}

As quantum  processors advance, the emergence of large-scale decentralized systems involving interacting quantum-enabled agents is on the horizon. Recent research efforts have explored quantum versions of Nash and correlated equilibria as solution concepts of strategic quantum interactions, but these approaches did not directly connect to decentralized adaptive setups where agents possess limited information. This paper delves into the dynamics of quantum-enabled agents within decentralized systems that employ no-regret algorithms to update their behaviors over time. Specifically, we investigate two-player quantum zero-sum games and polymatrix quantum zero-sum games, showing that no-regret algorithms converge to separable quantum Nash equilibria in time-average. In the case of general multi-player quantum games, our work leads to a novel solution concept, that of the {separable} quantum coarse correlated equilibria (QCCE), as the convergent outcome of the time-averaged behavior no-regret algorithms, offering a natural solution concept for decentralized quantum systems. 
Finally, we show that computing QCCEs can be formulated as a semidefinite program {and establish the existence of entangled (i.e., non-separable) QCCEs, which are unlearnable via the current paradigm of no-regret learning.}

\end{abstract}

\section{Introduction}

As quantum computing reaches maturity and quantum computing processors become more affordable and scalable, large-scale systems with interacting   quantum-enabled agents will become more commonplace. Quantum games offer a  powerful framework to predict the behavior  and guide the design  of such systems  \cite{gutoski2007toward,bostanci2022quantum,2012QSGT,wu2004new,khan2018quantum}. In a quantum game, agents can process and exchange quantum information, and their utilities are determined by performing  a measurement on a quantum state that is shared among all agents.

A significant portion of quantum game literature studies well-known games such as the Prisoner's Dilemma~\cite{eisert2000quantum} and Battle of the Sexes~\cite{marinatto2000quantum}, aiming to uncover potential advantages of using   quantum strategies when compared to classical ones.  Another significant research avenue centers on identifying suitable solution concepts for quantum games, which correspond  to   system states that exhibit stability against unilateral player deviations and are   collectively referred to as quantum equilibria. In particular, two notions of quantum equilibria have been  studied: quantum Nash equilibria and quantum correlated equilibria~\cite{jain2009parallel, bostanci2022quantum, 2012QSGT, 2013QCE}. Nevertheless, computing  quantum Nash equilibria is intractable \cite{bostanci2022quantum}, casting doubt on their suitability as a viable solution concept. Indeed,  in view of  the   hardness of computing quantum equilibria, how are agents expected to  reach such states?  To make matters worse, the hardness result holds even in settings   where agent's  utilities are known, an unrealistic assumption for large-scale decentralized~systems.
 %


 A more pragmatic setup is to consider that agents interact with each other over a series of rounds and  they have the opportunity to improve their strategies over time based on the outcomes of previous interactions. 
One established method to enable this dynamic learning process is \emph{no-regret learning}, where  agents update their strategies using an algorithm that meets a   natural benchmark; it performs as well as the best fixed strategy in hindsight. This leads to  the following key question which we seek to answer in this paper:

 \begin{center}
\parbox[c]{400pt}{\centering {\em For which classes of quantum games can  agents reach equilibria using no-regret learning? What types of equilibria do they arrive at? }
}
\end{center}


In the realm of classical normal-form games, where strategies are probability simplex vectors that capture classical randomness over a finite set of  actions,  it is well understood that  {\em no-regret learning converges to equilibrium states} \cite{cesa2006prediction,nisan2007algorithmic}. However, the type of equilibrium and notion of convergence depends on the specific setting and underlying applications. Comparatively, the study of no-regret learning in quantum games is in its infancy \cite{jain2009parallel,lotidis2023learning,jain2022matrix,lin2023quantum}. Our main  goal in this work is to   develop distributed  algorithms for learning in quantum games  and explore what types of equilibrium solutions emerge across different game classes.


\paragraph{Model, approach,  and  contributions.}
In this paper we focus on  a model of quantum games which is a natural extension of prior models, while still being amenable to no-regret learning. Formally, we focus on non-interactive quantum games, where
each player $i$ controls a quantum register $\mathcal{H}_i$ and has as their strategy a density matrix $\rho_i \in {\mathrm{D}}(\mathcal{H}_i)$. The joint strategy is given by a state $\rho \in {\mathrm{D}}(\bigotimes_i \mathcal{H}_i)$, and the payoff of the $i$-th player is given by
the expected value of an observable $R_i$ on the joint state, i.e.,
\begin{equation*}
 \label{QG:_multilinear}
\tag{QG}
	u_i(\rho) = \Tr(\rho R_i).
\end{equation*}
A \ref{QG:_multilinear} is called zero-sum if   players' payoffs  add up to zero.  More generally, we  also consider polymatrix~\ref{QG:_multilinear}s, where there are $k$ players and each player is situated at a node within an undirected graph. Every player engages in two-player  \ref{QG:_multilinear}s with each of their neighboring agents, employing a single state $\rho_i \in \mathrm{D}(\mathcal{H}_i)$ to participate in all games with their neighbors, and their individual payoff is  the cumulative payoff  earned across all these games. 


To investigate learning in quantum games, we draw inspiration from insights derived from no-regret learning in classical games. Within this line of research, we single out two important results.  Firstly,  no-external-regret learning gives rise to  decentralized algorithms that converge in the time-average sense to \emph{coarse correlated equilibria} \cite{nisan2007algorithmic} in arbitrary normal-form games. 
 Secondly, no-external-regret learners converge in the time-average sense to \emph{Nash equilibria} in both two-player zero-sum games~\cite{cesa2006prediction,freund1999adaptive} and also polymatrix (globally) zero-sum games~\cite{cai2011minmax,daskalakis2009network}.  
 
All  these results can be unified within the ${\bf \Phi}$-regret framework  \cite{greenwald2003general}.  The benchmark  of no-$\mathbf{\Phi}$-regret arises by  allowing agents to  deviate   from an action  $x$ to  $\phi(x)$, where $\phi$ is an admissible deviation mapping in a family $\mathbf{\Phi}$ of linear deviation maps. A unifying result of noteworthy relevance to our work is that  players using a no-$\mathbf{\Phi}$-regret algorithm  converge to a corresponding notion of $\mathbf{\Phi}$-equilibria in classical normal-form games~\cite{greenwald2003general}.

Our  results in this work show that  all aforementioned results for  no-regret learning in classical normal-form games  carry over to the quantum regime.  Specifically, in this work:
\begin{itemize}
\item We introduce the notion of quantum $\mathbf{\Phi}$-equilibria (\ref{QPhiE}) for any admissible family of quantum deviation maps ${\bf \Phi}$. Notably, 
the well-explored concepts of quantum Nash \eqref{QNE} and quantum correlated equilibria \eqref{QCE} both emerge as specific instances within this broader framework.
\item For any \ref{QG:_multilinear}, we show that no-$\mathbf{\Phi}$-regret learning converges to \ref{QPhiE}. Moreover, we show that   the set of \emph{separable} quantum coarse correlated equilibria (\ref{QCCE})  coincides with the limit points of the time-averaged achieved through  no-external-regret algorithms. \edits{On the other hand, \emph{entangled} \ref{QCCE}s cannot be reached through the current paradigm of learning in games, and we demonstrate that this class of equilibria is not vacuous by constructing examples of entangled \ref{QCCE}s in Appendix \ref{sec:maxent}.}
\item For two-player quantum zero-sum games  and polymatrix quantum zero-sum games, we show that the limit points of no-external-regret algorithms are separable \ref{QNE}s.
\end{itemize}


%

\paragraph{Related work.} Research on no-regret learning in quantum games is  relatively limited. In a pioneering work, \cite{jain2009parallel} focused on  Matrix Multiplicative Weights Update (MMWU),  a matrix extension of the widely-used Multiplicative Weights Update algorithm \cite{arora2012multiplicative}. Their research focused on two-player quantum zero-sum games, demonstrating convergence in a time-average sense to the set of~\ref{QNE}s. Our results   provide an alternative, simpler proof of that result, which   furthermore  holds for any no-external-regret algorithm. 

More recently, \cite{jain2022matrix} and later on \cite{lotidis2023learning} studied   continuous-time variants of MMWU and variants thereof. They demonstrated  a  cyclic behavior known as Poincaré recurrence in the dynamics, a phenomenon reminiscent of classical results indicating that regret minimization alone is insufficient for last-iterate equilibrium convergence, e.g. see~\cite{piliouras2014optimization, boone2019darwin, mertikopoulos2018cycles}.
 Beyond the zero-sum setting, \cite{lin2023quantum} studies the continuous and discrete variants of a linear version of MMWU in quantum potential games, showing that players' utilities strictly increase when using these~algorithms.

%

The particular context of learning within \ref{QG:_multilinear}s  investigated in this study can be regarded as a specific instance of the broader framework of learning in convex games, as outlined in works like \cite{gordon2008no} and \cite{stoltz2007learning}. Consequently, it becomes essential to elucidate the applicability of these general findings to our specific setting.

In particular,  \cite{gordon2008no} studies $\mathbf{\Phi}$-regret minimization in general  convex games and provides  a   template for designing no-$\mathbf{\Phi}$-regret algorithms, which entails that:  %
\begin{enumerate}
    \item The set of transformations $\mathbf{\Phi}$ is a \emph{reproducing kernel Hilbert space} (RKHS).
    \item We have access to an algorithm $\mathcal{A}'$ which computes approximate fixed points of any $\phi\in\mathbf{\Phi}$.
    \item We have access to an algorithm $\mathcal{A}''$ for no-regret learning in the setting where actions correspond to choosing a deviation $\phi \in \mathbf{\Phi}$.
\end{enumerate}

In terms of using this framework for no-regret learning in \ref{QG:_multilinear}s, the  
 most restrictive assumption is the third one.  An an example, in the case where $\mathbf{\Phi}$ is the set of all completely positive trace-preserving maps (i.e.,   linear maps that  transform  valid quantum states to valid quantum states),  the third assumption necessitates the existence of a no-regret algorithm for learning over  the domain of completely positive, trace-preserving (CPTP) maps. 
To the best of our knowledge, such an algorithm is not available, and obtaining one is the focus of ongoing work. 



Finally, \cite{stoltz2007learning} studies no-internal-regret convergence in convex games, however, their algorithm is not practically applicable as the time and space requirements
grow exponentially with the number of timesteps. In contrast, our approach is efficiently computable and in Section \ref{sec:experiments} we complement our theoretical results with related experiments.

\section{Quantum Games, Equilibria and Online Optimization}
\label{sec:2}

In this section, we introduce a broad formulation of quantum games and study non-commutative analogues of classical equilibrium concepts in such games,
before turning our attention to the equilibria we can learn and how to learn them in the subsequent sections.

\textbf{Quantum preliminaries.} 
A $d$-dimensional quantum register is mathematically described as the set of unit vectors  in a $d$-dimensional Hilbert space $\mathcal{H}.$
The \emph{state} of a qudit quantum  register $ \mathcal{H}$ is represented by a \emph{density matrix}, i.e.,  a $d\times d$ Hermitian positive semidefinite matrix with trace equal to $1$. A qudit is the unit of quantum information described by a superposition of $d$ states. If the number of states $d$ is equal to two then it is referred to as a qubit. The state space of a quantum register $\mathcal{H}$ is denoted  by  $\mathrm{D}(\mathcal{H})$.

One mathematical formalism of the process of measuring a  quantum system is the positive-operator-valued measure (POVM),  defined as a set of positive semidefinite operators $\{P_i\}_{i=1}^m$ such that $\sum_{i=1}^mP_i=\mathbbm{1}_\mathcal{H}$, where $\mathbbm{1}_\mathcal{H}$ is the identity matrix on $\mathcal{H}$. If the quantum  register  $\mathcal{H}$ is in a state described by density matrix $\rho\in \mathrm{D}(\mathcal{H})$, upon performing the measurement $\{P_i\}_{i=1}^m$ we get the outcome $i$ with probability $\langle P_i, \rho \rangle$,
where $\langle A, B\rangle = \Tr(A^\dag B)$ is the \emph{Hilbert-Schmidt inner product} defined on  the linear space of Hermitian matrices. A POVM can also be seen as a collection of observables, each corresponding to a Hermitian operator. In this paper we focus on the POVM formalism for quantum measurement, but there are other formalisms in the literature which we defer to Appendix \ref{appsec:prelims} for completeness. 

Given a finite-dimensional Hilbert space $\mathcal{H}=\mathbb{C}^n$, we denote by $\text{L}(\mathcal{H})$ the set of linear operators acting on $\mathcal{H}$, i.e.,   the set of all $n\times n$ complex matrices over $\mathcal{H}$. When two quantum registers with associated spaces $\mathcal{A}$ and $\mathcal{B} $ of dimension $n$ and $m$ respectively are considered as a joint quantum register, the associated state  space is given by the density operators  on the tensor product space, i.e., $D(\mathcal{A}\otimes \mathcal{B})$. A linear operator that maps matrices to matrices, i.e.,  a mapping  $\Theta:\mathrm{L}(\mathcal{B}) \to \mathrm{L}(\mathcal{A})$, is called a {\em super-operator}. The set of admissible super operators is equivalently the set of completely positive and trace preserving (CPTP) maps. The adjoint super-operator $\Theta^\dagger:\mathrm{L}(\mathcal{A}) \to \mathrm{L}(\mathcal{B})$  is uniquely determined by the equation
$    \langle A, \Theta(B)\rangle = \langle \Theta^\dagger(A), B\rangle
$.  A super-operator $\Theta:\mathrm{L}(\mathcal{B})\to\mathrm{L}(\mathcal{A})$ is    {\em positive} if it maps PSD matrices  to PSD matrices. 
There exists a  linear bijection between  matrices $R\in \mathrm{L}(\mathcal{A}\otimes\mathcal{B})$ and super-operators $\Theta:\mathrm{L}(\mathcal{B})\to\mathrm{L}(\mathcal{A})$ known as the {\em Choi-Jamio\l{}kowski isomorphism}. Specifically, for a super-operator $\Theta$  its {\em Choi matrix}~is:
\begin{equation}\label{CJ}
    C_\Theta= \sum_{1\leq i,j\leq m} \Theta(E_{i,j}) \otimes E_{i,j}\in \mathrm{L}(\mathcal{A}\otimes\mathcal{B}),
\end{equation}
where $\{E_{i,j}\}_{i,j=1}^m$ is the standard orthonormal basis of $\mathrm{L}(\mathcal{B}) = \mathbb{C}^{m\times m}$. Conversely, given an operator $R=\sum_{1\le i,j\le m}A_{i,j}\otimes E_{i,j}\in \mathrm{L}(\mathcal{A}\otimes\mathcal{B})$, we can define $\Theta_R:\mathrm{L}(\mathcal{B})\to\mathrm{L}(\mathcal{A})$ by setting $\Theta_R(E_{i,j})=A_{i,j}$ from which it easily follows that $C_{\Theta_R}=R$. Explicitly, we~have
\begin{equation}\label{eqn:superoperator}
    \Theta_R(B) = \mathrm{Tr}_\mathcal{B} (R(\mathbbm{1}_\mathcal{A}\otimes B^\top)),
    \end{equation}
    where the partial trace
    $ \mathrm{Tr}_\mathcal{B}:{\mathrm{L}}(\A \otimes \B)\to {\mathrm{L}}(\A)$
    is the {\em unique} function
 that satisfies:
\begin{equation*}\label{basic:ptrace}
\mathrm{Tr}_\mathcal{B}(A\otimes B)=A\Tr(B), \  \forall A, B.
\end{equation*}
Moreover, the  adjoint map is $\mathrm{Tr}_\mathcal{B}^\dagger(A)=A\otimes \mathbbm{1}_\B$.
Lastly,  a superoperator $\Theta$ is completely positive (i.e., $\mathbbm{1}_m\otimes \Theta$ is positive for all $m\in \mathbb{N}$) if and only if the Choi matrix of $\Theta$ is positive semidefinite. In particular, if the Choi matrix of the super-operator $\Theta$ is PSD, it follows that $\Theta$ is positive. 


Finally, if a state $\rho  \in \mathrm{D}(\bigotimes_i \mathcal{H}_i)$ can be written as a convex combination of product states,  i.e., 
\begin{equation}
\label{def:sep}
    \rho = \sum_j \lambda_j \bigotimes_i \rho_{i,j}
\end{equation} 
\edits{where $\lambda_j \geq 0 \ \forall j$, $\sum_j \lambda_j = 1$, and $\rho_{i,j} \in \mathcal{H}_i \ \forall j$, then it is called \emph{separable}. A separable state can be interpreted as a classical probability distribution over the product states $\rho_j := \bigotimes_i \rho_{i,j}$. A state that is not separable is said to be \emph{entangled}.} 

\subsection{Quantum games}\label{subsec:quantum_games}

In a  \ref{QG:_multilinear}, there are $k$ players and each player $i$ has register $\mathcal{H}_i$ and selects a density matrix $\rho_i \in {\mathrm{D}}(\mathcal{H}_i)$. A joint strategy is given by a joint state $\rho \in \mathrm{D}(\bigotimes_i \mathcal{H}_i)$. 
Each player has an observable $R_i=\sum_{m_i} m_i P_{m_i}$ and thus their utility function is the (multilinear) expected value of the observable $R_i$ on the joint state, i.e.,
\begin{equation}
\label{QG:_multilinear}\tag{QG}
	u_i(\rho) = \Tr(\rho R_i)
\end{equation}
for some Hermitian $R_i \in \text{L}(\otimes_i \mathcal{H}_i)$. We henceforth refer to $R_i$ as player $i$'s utility tensor.


It is useful to note that if $\rho = \rho_i \otimes \rho_{-i}$ for some $i \in [k]$, $\rho_{-i} \in \mathrm{D}(\bigotimes_{i' \neq i} \mathcal{H}_{i'})$, then we can also write the utility \eqref{QG:_multilinear} in the alternative form
\begin{equation}\label{multutility}
    u_i(\rho) = \Tr(\rho R_i) = \innerprod{\rho_i}{\Theta_i(\rho_{-i}^\top)}
\end{equation}
where $\Theta_i: {\mathrm L}(\bigotimes_{i' \neq i} \H_{i'}) \to {\mathrm L}(\H_i)$ and $R_i \in {\mathrm L}(\bigotimes_{i'} \H_{i'}) = {\mathrm L}(\H_i \otimes (\bigotimes_{i' \neq i} \H_{i'}))$ are related via the Choi-Jamio\l{}kowski isomorphism. This is because, by \eqref{eqn:superoperator}, \[
    \innerprod{\rho_i}{\Theta_i(\rho_{-i}^\top)} 
    = \innerprod{\rho_i}{\Tr_\B(R(\mathbbm{1}_\A \otimes \rho_{-i}))}
    = \innerprod{\rho_i \otimes \mathbbm{1}_\B}{R(\mathbbm{1}_\A \otimes \rho_{-i})}
    =\Tr( R(\rho_i \otimes \rho_{-i})).
\]

An important special case of the \ref{QG:_multilinear} framework was introduced by Zhang \cite{2012QSGT} in an attempt to explore whether quantum resources provide advantages in classical games. 
In Zhang's model, the strategies of each player are encoded by a Hilbert space $\H_i$ with an orthonormal basis $\ket{s_i}, s_i\in S_i$ and the probability of strategy profile $s \in \times_i S_i$  is given by $\Tr(\ketbra{s}{s}\rho)$, where the shared state for all players is $\rho\in D(\otimes_i \H_i)$. Consequently the expected utility of the $i$-th player is given by 
\[
u_i(\rho) = \sum_{s\in S} \Tr(\ketbra{s}{s}\rho) u_i(s) = \Tr(\sum_{s\in S} u_i(s) \ketbra{s}{s}\rho),
\]
which is  the form of \ref{QG:_multilinear} when restricted to diagonal   utility tensors, i.e., $R_i = \sum_{s \in S}  u_i(s) \ketbra{s}{s}$. 

Moreover, the \ref{QG:_multilinear} framework also captures the work of \cite{jain2009parallel, ickstadt2022semidefinite}, who consider a related model of non-interactive quantum games wherein players each control a quantum register. They each independently prepare a quantum state in the register they hold, which are subsequently sent to a referee who performs a joint measurement to determine a real payoff to each player.  
\edits{Crucially, the \ref{QG:_multilinear} framework can be seen as the first stage of the more general quantum game framework found in \cite{gutoski2007toward, bostanci2022quantum}. In this broad framework, quantum strategies are defined over $n$ rounds of interaction. Each round is characterized by an input and output space, as well as admissible mappings between rounds which can capture quantum memory. Finally, at the last memory state, a joint measurement is made to determine the payoffs. We focus on a specialization of this framework, where players select a strategy without prior communication or entanglement, and a measurement is performed using the joint state of the players. This allows us to study a closer quantum analogy to the framework of classical, simultaneous \emph{non-cooperative} games.}

In our setting, the extended utility function of \ref{QG:_multilinear}, $u_i(\rho) = \Tr(\rho R_i)$, can be interpreted as the expected utility if the product state (i.e., strategy profile) to be played is separable. 
On the other hand, an entangled state $\rho$ can only be played by a central agent who 
plays on behalf of all the players.
We explore this problem in more detail in Appendix \ref{sec:maxent}.

\subsection{Various notions of quantum equilibria}\label{subsec:qeq}

In prior works in quantum game theory, classical notions of equilibria have been studied in the quantum context, and several classical results have been shown to have non-commutative analogues. In this section, we recall some of these notions and introduce the quantum coarse correlated equilibrium (\ref{QCCE}), which we first write in terms of deviation mappings and later reformulate in terms of partial traces.

A seminal equilibrium concept in classical game theory is the Nash equilibrium.
In our formulation \eqref{QG:_multilinear}, 
we have the notion of an $\epsilon$-quantum Nash equilibrium ($\epsilon$-QNE),
from which players can only make utility gains of $\leq \epsilon$ by deviating. This formulation of the quantum Nash equilibrium has already been studied in several works \cite{bostanci2022quantum,jain2009parallel,ickstadt2022semidefinite} and we repeat it using our notation for completeness.

\begin{definition}
	A product state $\rho = \bigotimes_i \rho_i \in \mathrm{D}(\otimes_i \H_i)$ is a {quantum Nash equilibrium} (QNE) if 
        \begin{equation}
	\label{QNE}\tag{QNE}
		u_i (\rho) 
        \geq 
        u_i (\rho_i' \otimes \rho_{-i} )
        \quad
        \forall \ i \in [k], \; \rho_i' \in \mathrm{D}(\H_i),
	\end{equation}
        where $\rho_{-i} := \bigotimes_{j \neq i} \rho_{j}$. 
        Moreover $\rho$ is called an $\epsilon$-approximate quantum Nash equilibrium ($\epsilon$-QNE) if the inequality is satisfied up to an additive error of $\epsilon$.
\end{definition}

The \emph{exploitability} (see e.g. \cite{johanson2011accelerating}) of player $i$ at state $\rho$ is the maximum utility they can gain by deviating, and is given by
$
    \lambda_{\max}(\Theta_i(\rho_{-i})) - u_i(\rho) \geq 0.
$
This leads to an alternative characterization of \ref{QNE}s as
the product states that have
zero total exploitability, i.e.,
\begin{equation}
    \rho \text{ is a QNE}
    \Longleftrightarrow
    \sum_{i=1}^k \Big(\lambda_{\max}(\Theta_i(\rho_{-i})) - u_i(\rho) \Big)
    =
    0.
\end{equation}


Clearly, \ref{QNE}s are states which are stable under unilateral player deviations. However, when we consider non-product states, then we need to formulate a better way to capture families of unilateral deviations permitted within the quantum mechanics framework.  These deviations are captured by quantum channels and mathematically formalized by completely positive, trace-preserving (CPTP) maps, which are the family of linear mappings from density matrices to density matrices. This leads to the following~definition: 



\begin{definition}
     Consider a family of CPTP maps $\mathbf{\Phi}=\{\Phi_i\}_{i=1}^k$. A state $\rho\in \mathrm{D}(\otimes_i \H_i)$ is called a quantum $\mathbf{\Phi}$-equilibrium if
	\begin{equation}
	\label{QPhiE}\tag{Q$\mathbf{\Phi}$E}
		u_i(\rho) \geq u_i((\phi_i \otimes \id_{-i})(\rho)) \quad  \forall \ i, \ \phi_i\in \Phi_i.
	\end{equation}
	Moreover  $\rho$ is an $\epsilon$-\emph{quantum $\mathbf{\Phi}$-equilibrium} ($\epsilon$-\ref{QPhiE}) if the inequality is satisfied up to an additive error of $\epsilon$.
\end{definition}

Specializing this formulation and allowing for all possible CPTP maps, we arrive at the notion of quantum correlated equilibria (\ref{QCE}), first defined in \cite{2012QSGT} and analyzed further in \cite{2013QCE}. Concretely, we have~that:


\medskip
\begin{definition}
    A state $\rho\in \mathrm{D}(\otimes_i \H_i)$ is called a quantum correlated equilibrium if
	\begin{equation}
	\label{QCE}\tag{QCE}
		u_i(\rho) \geq u_i((\phi_i \otimes \id_{-i})(\rho))
	\end{equation}
	where $\phi_i : \mathrm{L}(\H_i) \rightarrow \mathrm{L}(\H_i)$ is a CPTP map on player $i$'s subsystem. Moreover  $\rho$ is an $\epsilon$-\emph{quantum correlated equilibrium} ($\epsilon$-QCE) if the inequality is satisfied up to an additive error of $\epsilon$.
\end{definition}

As a second instantiation of $\mathbf{\Phi}$-equilibria, we consider the set of "constant maps", which in our setting corresponds to replacement channels. This leads to a new notion of quantum equilibria, described below:



\begin{definition}
        A state $\rho\in \mathrm{D}(\otimes_i \H_i)$  is called a {quantum coarse correlated equilibrium} if 
 \begin{equation}\label{QCCE}
    u_i (\rho) \geq u_i ((\phi_i \otimes \id_{-i})(\rho))\tag{QCCE}
\end{equation}
for all players $ i \in [k]$ and all replacement channels 
\begin{equation}\label{rchannel}
\phi_i: \mathrm{L}(\H_i) \rightarrow \mathrm{L}(\H_i), \; X \mapsto \Tr(X) \rho_i', \quad\mathrm{where }\ \rho_i' \in \mathrm{D}(\H_i).
\end{equation}

\end{definition}

Equivalently, we can express the \ref{QCCE} definition in terms of partial traces. In particular, $\rho$ is a \ref{QCCE}~if
    \begin{equation}
	\label{eq:dev-QCCE}
		u_i (\rho) \geq u_i (\rho_i' \otimes \Tr_i \rho)
	\end{equation}
	for all players $ i \in [k]$ and $ \rho_i' \in\mathrm{D}(\H_i)$, where $\Tr_i : \mathrm{L}(\bigotimes_{i'} \H_{i'}) \rightarrow \mathrm{L}(\bigotimes_{i' \neq i} \H_{i'})$ is the partial trace with respect to player $i$'s subsystem.  Moreover, $\rho$ is  an $\epsilon$-approximate {quantum coarse correlated equilibrium} ($\epsilon$-QCCE) if the inequality is satisfied up to an additive error of $\epsilon$.
Finally, if a QCCE $\rho$ is a separable state (i.e., it can be expressed as a convex combination of product states), we call it a separable \ref{QCCE}. A proof of equivalence between \ref{QCCE} and \eqref{eq:dev-QCCE} can be found in Lemma \ref{lem:QCCE_equiv} of the Appendix.

By definition, a product \ref{QPhiE} is a \ref{QNE}.
Thus, the space of states considered matters greatly in our equilibrium definitions. For \ref{QG:_multilinear}s it turns out that, unlike the classical case, there is a space of states of interest between that \edits{the narrow class of product states and the broad class of (possibly entangled) general states, which is the space of separable states \eqref{def:sep}. Separable equilibria are an interesting class of equilibria to consider not only because of the significance of separable states in quantum theory, but also because separable states are the set of states reachable by the current paradigm of no-regret learning, in which a product state is played in each round of play and equilibria are obtained by taking a time average---i.e., a convex combination.

Nevertheless, entangled (i.e., non-separable) equilibria can exist. In particular, in  we provide explicit constructions of maximally-entangled \ref{QCCE}s in Appendix \ref{sec:maxent}.
}

\paragraph{Spectrahedral characterization of \ref{QCCE}s.}
Analogous to the classical setting where the set of CCEs can be described as the feasible space of a linear program \cite{roughgarden2016twenty}, the set of \ref{QCCE}s of a game can be described as the feasible space of a semidefinite program (SDP). Suppose that, for each $i \in [k]$, $\Theta_i: {\mathrm L}(\bigotimes_{i' \neq i} \H_{i'}) \to {\mathrm L}(\H_i)$ and $R_i \in {\mathrm L}(\bigotimes_{i'} \H_{i'}) = {\mathrm L}(\H_i \otimes (\bigotimes_{i' \neq i} \H_{i'}))$ are related via the Choi-Jamio\l{}kowski isomorphism. Then a density matrix $\rho^*$ is a \ref{QCCE} if and only~if

\begin{align*}
    &u_i (\rho^*) \geq u_i (\rho_i' \otimes \Tr_i \rho^*)
	&\forall i \in [k], \rho_i' \in \H_i \\
    \Leftrightarrow&
    \Tr(R_i \rho^*) \geq \Tr(R_i(\rho_i' \otimes \Tr_i \rho^*))
	&\forall i \in [k], \rho_i' \in \H_i \\
    \Leftrightarrow&
    \Tr(R_i \rho^*) \geq \max_{\rho_i' \in \H_i}  \innerprod{\rho_i'}{\Theta_i({(\Tr_i \rho^*)}^\top)}
	&\forall i \in [k] \\
    \Leftrightarrow&
	\Tr(R_i \rho^*) \geq \lmax(\Theta_i({(\Tr_i \rho^*)}^\top))
	&\forall i \in [k] \\
    \Leftrightarrow&
 	\Tr(R_i \rho^*)I_i - \Theta_i({(\Tr_i \rho^*)}^\top) \succeq 0
	&\forall i \in [k], 
\end{align*}


so
\[
	\textrm{\ref{QCCE}s} = \{ \rho^* \in \mathrm{\H}: \Tr(R_i \rho^*)I - \Theta_i({(\Tr_i \rho^*)}^\top) \succeq 0 \,\ \ 
	\forall i \in [k],
	\;
	\Tr \rho^* = 1,
	\;
	\rho^* \succeq 0
	\}.
\]

These conic inequality constraints can be combined into a single, block-diagonal linear matrix inequality in terms of the entries of $\rho^*$, giving us an SDP characterization of the set of \ref{QCCE}s of a given game. Hence, the set of \ref{QCCE}s is a spectrahedron, mirroring the classical result that the set of CCEs is a polyhedron.


\paragraph{$\bf{\Phi}$-equilibria in classical games.}
To give context for quantum $\mathbf{\Phi}$-equilibria to readers who may be unfamiliar with with the concept, we shall express well-known classical equilibria in the $\mathbf{\Phi}$-equilibria framework.
A  normal-form game  consists of a set of players ${\cal N}=\{1,\ldots,k\}$ where player $i$ may select from a finite set of actions or pure strategies~${S}_i$. Additionally, each player has a payoff function $u_i: {S}\equiv 
\prod_{i} 
{S}_i \to \mathbb{R}$ assigned over pure strategy profiles $s=(s_1,\ldots, s_k)$. A 
joint strategy
$p\in \Delta(\times_iS_i)$ is a probability distribution over the 
 space $\times_i S_i$ of pure strategy profiles, where  $p(s_1,\ldots,s_k)$ is the probability the system is in the pure state $(s_1,\ldots,s_k)$.  The expected payoff of  the $i$-th player 
 given that the
joint strategy
$p$ is played is given by $u_i(p)=\sum_{s_i\in S_i: i \in {\cal N}}p(s_1,\ldots,s_k)u_i(s_1,\ldots, s_k),$

The analysis of normal-form games typically boils down to equilibrium analysis: studying what the players of the game eventually fall to as the best strategy to employ. The  most famous form of equilibrium is the Nash equilibrium \cite{nash1951games}, which is however intractable to compute. 
Several alternative equilibrium concepts have been proposed, namely correlated equilibria (CE) \cite{aumann1974subjectivity} and coarse correlated equilibria (CCE) \cite{moulin1978strategically}.
All three of these equilibrium types  can be effectively unified and examined through the lens of $\mathbf{\Phi}$-equilibria~\cite{greenwald2003general}. 

Formally, let $\Phi_i$ be a family of deviations 
  for   agent $i$, i.e.,  for any $\phi_i\in \Phi_i$ we have that $\phi_i(\Delta(S_i))\subseteq \Delta(S_i)$, so for each $s'_i\in S_i$ the vector $\phi(s'_i)$ is  a distribution.  Then,  for any joint strategy  $p\in \Delta(\times_iS_i)$ and   $\phi_i\in \Phi_i$,  we define a new joint strategy   $(\phi_i\otimes \mathbb{I}_{-i})(p)$ that assigns to the strategy profile $(s_i, s_{-i})$ the probability
$\sum_{s'_i}p(s'_i, s_{-i}){\rm Prob}(s'_i\overset{\phi_i}{\to}s_i)$.
Setting ${\bf \Phi}=\{\Phi_i\}_{i=1}^k$,  the joint  strategy     $p$ is called  a ${\bf \Phi}$-equilibrium if 
\begin{equation}\tag{$\mathbf{\Phi}$-equilibrium} 
u_i(p)\geq u_i((\phi_i\otimes \mathbb{I}_{-i})(p)) \quad \forall \ i, \ \phi_i\in \Phi_i,
\end{equation} 
or explicitly:
\begin{equation}\label{explicit}
u_i(p)\ge \sum_{s_i, s_{-i}}u_i(s_i,s_{-i})\sum_{s'_i}p(s'_i,s_{-i}){\rm Prob}(s'_i\overset{\phi_i}{\to}s_i).
\end{equation}


Nash, correlated, and coarse correlated equilibria can all be seen as specific instances of ${\bf \Phi}$-equilibria through suitable choice of deviation mappings.  Concretely, correlated equilibria correspond to the case
where permissible deviations are   linear maps  $\phi_i$ that map distributions to distributions, i.e.,  $\phi_i(\Delta(S_i))\subseteq  \Delta(S_i)$. 
 By linearity, any  such map   can be written as $\phi_i(x)=A_ix$ , and since $\phi_i$ preserves the simplex  $A_i$ is entrywise nonnegative and column stochastic.
 On the other hand, coarse correlated equilibria correspond to the case where  the allowable deviations are the constant maps, i.e., the map  $ \Delta(S_i)$ to a single point $x_i  \in \Delta(S_i)$.  Finally, a joint strategy $p$ is a Nash equilibrium iff it is a CCE and a product distribution. 

\subsection{No-$\mathbf{\Phi}$-regret learning in quantum games}\label{sec:learningequilibriaquantum}

The concept of regret serves as a well-established measure in assessing the performance of online algorithms~\cite{cesa2006prediction,hazan2016introduction}. In the ensuing discussion, we introduce the notion of ${\bf \Phi}$-regret within the context of online linear optimization over the set of density matrices.

Let's consider an algorithm $\mathbf{A}$ that generates a sequence of iterates $\rho^t\in D(\mathcal{H})$. The ${\bf \Phi}$-regret benchmark compares  the cumulative utility achieved by the trajectory $\{\rho^t\}_{t=0}^T$ with  the best attainable utility when deviating from $\rho^t$ to $\phi(\rho^t)$ at each step, with $\phi\in \mathbf{\Phi}$ being a fixed deviation map, i.e., 
\begin{equation}
    \textrm{regret}^{T, \mathbf{\Phi}}(\mathbf{A})= \max_{\phi\in \mathbf{\Phi}} \sum_{t=1}^T \innerprod{\phi(\rho^t)}{R}-\sum_{t=1}^T \innerprod{\rho^t}{ R},
\end{equation}
An online algorithm $\mathbf{A}$ is called "no-${\bf \Phi}$-regret" if the normalized ${\bf \Phi}$-regret ${1\over T}   \textrm{regret}^{T,\mathbf{\Phi}}(\mathbf{A})$ tends towards zero as $T$ grows. In line with conventional terminology in the literature, an algorithm is deemed "no-external-regret" (or simply "no-regret," where context permits) when all admissible constant maps are considered as deviations. In this case, these deviations correspond to all replacement channels, as defined in~\eqref{rchannel}.

 The first example of a no-external-regret algorithm for online linear optimization over the set of density matrices is the Matrix Multiplicative Weights Update (MMWU) method, e.g. see \cite[Theorem 10]{kale2007efficient} and Algorithm~\ref{alg:mmwu}. MMWU is a widely applicable no-external-regret algorithm which has found applications for online optimization over the set of density matrices \cite{arora2012multiplicative,kale2007efficient,tsuda2005matrix}. 
Some specific applications include solving semidefinite programs (SDPs) \cite{arora2007combinatorial}, proving QIP=PSPACE \cite{jain2011qip} and spectral sparsification \cite{allen2015spectral}.

Furthermore, in \cite{allen2015spectral} it is shown how MMWU arises naturally as an instance of the Follow-the-Regularized-Leader framework where the regularizer is chosen to be the
entropy function. Based on this, they introduce the novel FTRL$_{\mathrm{exp}}$ framework for a different class of regularizers and provide corresponding regret bounds. It turns out that this class of algorithms is also no-external-regret in our setting.

\begin{algorithm}[h!]
\caption{Matrix Multiplicative Weights Update (MMWU)\label{alg:mmwu}}
\begin{algorithmic}
   \STATE {\bfseries Initialize} weight matrix $W_i^{1} = \mathbbm{1}_{d}$ and stepsize $\eta \leq \frac{1}{2}$.
   \FOR{$t=1\dots T$}
   \STATE {Play using {$\rho^t_i = \frac{W_i^t}{\Tr(W_i^t)}$}}
   \STATE {Update weight matrix $W_i^{t+1} = \exp(\eta \sum_{\tau=1}^t \Theta\left((\rho^{\tau}_{-i})^\top\right) )$}
   \ENDFOR
\end{algorithmic}
\end{algorithm}

 Beyond online optimization, in this work we  consider the  setup where  players in a game interact with each other over a series of rounds and  improve their strategies using a no-regret algorithm $\mathbf{A}$.   Let $\Phi_i$ be a set of deviation mappings for each agent and let ${\bf \Phi}=\{\Phi_i\}_{i=1}^k$. Recalling, that in the \ref{QG:_multilinear}  setup the payoffs are multilinear \eqref{multutility}, each player $i$'s regret for using an online algorithm $\mathbf{A}$ with deviations $\Phi_i$ is 
\[
\textrm{regret}_i^{T,\Phi_i}(\mathbf{A},{\bf \Phi})= \max_{\phi_i\in\Phi_i} \sum_{t=1}^T \innerprod{\phi_i(\rho^t_i)}{\Theta\left((\rho^{t}_{-i})^\top\right)}-\sum_{t=1}^T \innerprod{\rho^t_i}{\Theta\left((\rho^{t}_{-i})^\top\right)}.
\]
Moreover, $\mathbf{A}$ has the no-${\bf \Phi}$-regret property if 
${1\over T}   \textrm{regret}^{T}(\mathbf{A},{\bf \Phi})\to 0$.  In Theorem \ref{thm:_sep-QCCE=timeave}\ref{thm:_QCCEa} we show that for  any set of deviation mappings $\mathbf{\Phi} = \{\Phi_i\}_{i=1}^k$, the limit points of the time-averaged joint history of play $\left\{ \frac{1}{T} \sum_{t=1}^T \bigotimes_i \rho_i^t \right\}_T$ generated by no-$\mathbf{\Phi}$-regret algorithms are separable \ref{QPhiE}.

\paragraph{No $\mathbf{\Phi}$-regret in classical games.} In a classical normal-form game, joint strategies are distributions over pure action  profiles  and agents use deviation mappings   ${\bf \Phi}$ that are linear maps that map distributions to distributions. Moreover, as payoffs are multilinear in a normal-form game, each agent updates their strategy using an algorithm for online optimization over their corresponding probability simplex. 

A pivotal result highlighted in \cite{greenwald2003general} demonstrates the existence of no-${\bf \Phi}$-regret algorithms for any family $\mathbf{\Phi}$ of linear deviation maps. Crucially, there exists an interesting connection between no-$\mathbf{\Phi}$-regret and game theoretic equilibria. Specifically, the time-average behavior of players using a no-$\mathbf{\Phi}$-regret algorithm converges to the corresponding notion of $\mathbf{\Phi}$-equilibria in general normal-form games.

Furthermore, the significance of no-external-regret algorithms is underscored by the folk theorem, which posits that  if all players employ external regret-minimizing algorithms to select their strategies, the players’ time-average behavior converges to the set of  coarse correlated equilibria \cite{roughgarden2016twenty,hart2000simple}. In addition, in the setting of two-player and polymatrix zero-sum games, the product of 
the players' individual time-averaged strategies converges to the set of Nash equilibria \cite{freund1999adaptive, cesa2006prediction, cai2011minmax, daskalakis2009network}.

\section{No-Regret Learning in General Quantum Games}
In this section we study \ref{QG:_multilinear}s from the perspective of no-$\mathbf{\Phi}$-regret learning and  provide an analogue to the classical CCE convergence result for no-external-regret learning.

\noindent
\begin{theorem}[Main Theorem]\label{thm:_sep-QCCE=timeave}
    For any quantum game we have the following:
    \begin{enumerate}[(a)]
        \item \label{thm:_QCCEa} For any deviation mappings $\mathbf{\Phi} = \{\Phi_i\}_{i=1}^k$, the limit points of the time-averaged joint history of play $\left\{ \frac{1}{T} \sum_{t=1}^T \bigotimes_i \rho_i^t \right\}_T$ generated by no-$\mathbf{\Phi}$-regret algorithms are separable \ref{QPhiE}. In particular, if all players update their strategies with a no-$\mathbf{\Phi}$-regret algorithm that guarantees a time-averaged regret of $\leq \epsilon$ after $T$ timesteps, then the time-averaged joint history of play after $T$ timesteps is a separable $\epsilon$-\ref{QPhiE}.
        \item \label{thm:_QCCEb} For any separable \ref{QCCE} $\rho^*$, there exist no-external-regret algorithms for each player so that
        their time-averaged joint history of play $\left\{ \frac{1}{T} \sum_{t=1}^T \bigotimes_i \rho_i^t \right\}_T$ converges to~$\rho^*$.
    \end{enumerate}
\end{theorem}
In particular, Theorem \ref{thm:_sep-QCCE=timeave} implies that for any quantum game the set of separable \ref{QCCE}s is equal to the limit points of the time-averaged history $\overline{\rho}(T) := \frac{1}{T} \sum_{t=1}^T \rho^t$
of joint play of players using no-external-regret algorithms. \edits{We note that this is the best statement that can be written for learning \ref{QCCE}s since taking the time-averaged history of joint play can only ever yield separable states (due to the fact that at each round a product state is played). On the other hand, because entangled \ref{QCCE}s can exist (see Appendix \ref{sec:maxent}), this means that there exist \ref{QCCE}s that are unlearnable via the current paradigm of no-regret learning.
}

We dedicate the rest of this subsection to proving Theorem \ref{thm:_sep-QCCE=timeave}\edits{, as well deriving an explicit convergence rate to separable $\epsilon$-\ref{QCCE}s when MMWU (Algorithm \ref{alg:mmwu}) is used}.
\medskip
\begin{proof}[Proof of Theorem \ref{thm:_sep-QCCE=timeave}(a)]
    Suppose that after $T$ iterations of running no-$\mathbf{\Phi}$-regret algorithms, every player has $\mathbf{\Phi}$-regret $\leq \epsilon = \epsilon(T)$. Let $\rho^t = \bigotimes_{i=1}^k \rho_i^t$ denote the strategy profile (product distribution) at time $t$, and let $\overline{\rho} = \overline{\rho}(T) := \frac{1}{T} \sum_{t=1}^T \rho^t$ be the time-averaged history of these strategy profiles. (It is thus a classical probability distribution over product distributions.) 

    That player $i$ has $\mathbf{\Phi}$-regret $\leq \epsilon$ means that the player $i$'s realized time-averaged utility is at least the time-averaged utility obtained by applying the channel $\phi_i \otimes \id_{-i}$ for any $\phi_i \in \Phi_i$, i.e.,

\begin{equation}	
\label{eqn:_QCCE_regret_bound}
		\frac{1}{T} \sum_t \Tr(R_i \Big( \bigotimes_i \rho_i^t \Big))
		\geq
		\frac{1}{T} 
		\sum_t
		\Tr(
                R_i
			(\phi_i \otimes \id_{-i})
                    \Big( \bigotimes_i \rho_i^t \Big)
		)
		-
		\epsilon
		\qquad
		\forall \phi_i \in \Phi_i.
\end{equation}
	But player $i$'s realized time-averaged utility is simply the (expected) utility he would get in one round if all the players play according to the time-averaged joint history $\rho^*$ since
	\[
		\frac{1}{T} \sum_t \Tr(R_i \Big( \bigotimes_i \rho_i^t \Big))
		=
		\Tr(
            		R_i
            		\left(
            			\frac{1}{T} \sum_t \bigotimes_i \rho_i^t
            		\right)
            	)	
            	=
            	\Tr(R_i \overline{\rho})
            	=
            	u_i (\overline{\rho}),
	\]
	while on the right-hand side of \eqref{eqn:_QCCE_regret_bound} we have that
	\[
		\frac{1}{T} 
		\sum_t
		\Tr(
                R_i
			(\phi_i \otimes \id_{-i})
                    \Big( \bigotimes_i \rho_i^t \Big)
		)
		=
		\Tr(R_i 
                (\phi_i \otimes \id_{-i})
                \left(
                    \frac{1}{T} \sum_t \bigotimes_i \rho_i^t
                \right)
                )
		=
            u_i((\phi_i \otimes \id_{-i})(\overline{\rho})).
	\]
	
	Thus we have from the regret bound \eqref{eqn:_QCCE_regret_bound} written for each player $i$ that

	\[
		u_i(\overline{\rho}(T))
		\geq
		  u_i((\phi_i \otimes \id_{-i})(\overline{\rho}(T)))
		-
		\epsilon(T)
		\quad
		\forall i \in [k], \; \phi_i \in \Phi_i,
	\]
	and taking the limit of these equations as $T \rightarrow \infty$ we get that
	\[
		\lim_{T \rightarrow \infty} 
		u_i(\overline{\rho}(T))
		\geq
		\lim_{T \rightarrow \infty} 
		u_i((\phi_i \otimes \id_{-i})(\overline{\rho}(T)))
		\quad
		\forall i \in [k], \; \phi_i \in \Phi_i.
	\]
	
	Finally, where $\rho^* := \lim_{m \rightarrow \infty} \overline{\rho}(T_m)$ is any limit point of the no-regret play (here $(T_m)_{m=1}^\infty$ is a subsequence of $\mathbb{N}$ for which the subsequence $(\overline{\rho}(T_m))_{m=1}^\infty$ converges), we have from the continuity of the payoff functions $u_i$ and the quantum channels $\phi_i \otimes \id_{-i}$ for all $i \in [k], \; \phi_i \in \Phi_i$ that
	\[
		u_i(\rho^*)
		=
		\lim_{m \rightarrow \infty} u_i(\overline{\rho}(T_m))
		\geq
		\lim_{m \rightarrow \infty} 
		u_i \left(
			(\phi_i \otimes \id_{-i})(\overline{\rho}(T_m))
		\right)
		=
		u_i((\phi_i \otimes \id_{-i})(\rho^*))
		\quad 
		\forall i \in [k], \; \phi_i \in \Phi_i,
	\]
	i.e. $\rho^*$ is a \ref{QPhiE}. Since the set of separable states is compact and each $\overline{\rho}(T)$ is separable, so the limit point $\rho^*$ is itself separable, and hence a separable \ref{QPhiE}.
\end{proof}

\begin{proof}[Proof of Theorem \ref{thm:_sep-QCCE=timeave}(b)]
	Let $\rho^* = \sum_{j=1}^m \lambda_j \bigotimes_i \rho_{ij}$ be a separable \ref{QCCE}. We can create a sequence of play $\left\{ \bigotimes_{i} \rho_i^t \right\}_t$ that converges to $\rho^*$ in terms of its time-averaged joint history as $t \rightarrow \infty$, i.e.
	\[
		\lim_{T \rightarrow \infty} \overline{\rho}(T) = \rho^*
		\qquad
		\text{where}
		\qquad
		\overline{\rho}(T) := \frac{1}{T} \sum_{t=1}^T \bigotimes_i \rho_i^t,
	\]
	by simply creating a sequence whose terms are in $[m]$ and whose frequencies converge to the distribution $(\lambda_j)_j$, then playing the product distribution $\bigotimes_i \rho_{ij}$ in time $t$ if the $t$-th element of the sequence is $i$.

	Now define the function $f : \bigotimes_i \H_i \rightarrow \R$ such that
	\[
	f(\rho) := \max_i \sup_{\rho_i' \in \H_i} \big[ u_i(\rho_i' \otimes \Tr_i \rho) - u_i (\rho) \big].
	\]
	The function $f$ is continuous by Lemma \ref{lem:_sup_unifcont_cont} since the functions $h_i: \left(\bigotimes_{i'} \H_{i'}\right) \times \H_i$, $h_i(\rho, \rho_i') := u_i(\rho_i' \otimes \Tr_i \rho)$ are continuous $\forall i$, and so we have that $\lim_{T \rightarrow \infty} f(\overline{\rho}(T))$ exists and 
	\[
		\lim_{T \rightarrow \infty} f(\overline{\rho}(T))
		=
		f \left(\lim_{T \rightarrow \infty} \overline{\rho}(T) \right)
		=
		f(\rho^*)
		\leq
		0,
	\]
	with the last inequality due to the fact that $\rho^*$ is a \ref{QCCE}.
	
    But the value of the function $f(\overline{\rho}(T))$ is equal to the maximal time-averaged regret that any player obtains up till time $T$, since
	\begin{equation*}
	\begin{split}
		\textrm{regret}_i^T=&
		\sup_{\rho_i' \in \H_i}
			\frac{1}{T}
			\sum_{t=1}^T
    				\left[
    					u_i \left(\rho_i' \otimes \left( \bigotimes_{i' \neq i} \rho_{i'}^t \right) \right)
					-
					u_i \left( \bigotimes_{i'} \rho_{i'}^t \right)
    				\right]
		\\
		=&
		\sup_{\rho_i' \in \H_i}
			\left[
				u_i \left( 
					\rho_i'
					\otimes
					\left(
						\frac{1}{T}
						\sum_{t=1}^T
							\bigotimes_{i' \neq i} \rho_{i'}^t
					\right)
				\right)
				-
				u_i \left(
					\frac{1}{T}
					\sum_{t=1}^T
						\bigotimes_{i'} \rho_{i'}^t 
				\right)
			\right]
		\\
		=&
		\sup_{\rho_i' \in \H_i}
			 \big[ u_i(\rho_i' \otimes \Tr_i \overline{\rho}(T)) - u_i (\overline{\rho}(T)) \big].
	\end{split}
	\end{equation*}
	
	Thus the maximal regret that any player obtains up till time $T$ converges to a value $\leq 0$ as $T \rightarrow \infty$, i.e. the sequence of play is no-regret.

    Finally, the no-regret algorithms that converge in time-averaged joint history to $\rho^*$ can be defined as follows: for player $i$, for time $t=1$ play $\rho_i^1$ (from the above sequence of play), and for time $t \geq 2$ check if all other players $i'$ have played according to the sequence $(\rho_{i'}^\tau)_\tau$ for all $\tau < t$. If YES, continue playing $\rho_i^t$; if NO, default to a guaranteed no-regret algorithm (e.g., MMWU) for all future time.
\end{proof}

Finally, since we have explicit no-external regret algorithms for learning in quantum games, we can give an explicit convergence rate to \ref{QCCE}s obtained by  all players using one such algorithm (MMWU):

\begin{remark}\label{remark:generalQG}
For a quantum game, if all utilities are in $[-1,1]$ and all players use MMWU with stepsize $\eta = \frac{\epsilon}{2}$ to update their strategies, then for any $\epsilon \leq 2$ 
    their time-averaged joint history of play $\frac{1}{T} \sum_{t=1}^T \bigotimes_i \rho_i^t$ after $T = \frac{4 \ln n}{\epsilon^2}$ steps is a separable~$\epsilon$-\ref{QCCE}. 
\end{remark}

 To see why this is the case, recall from Section \ref{subsec:quantum_games} that the utility of player $i$ at time $t$ can be written as 
    $u_i(\rho^t) = \Tr(\rho^t R_i) = \innerprod{\rho^t_i}{\Theta_i({\rho^t_{-i}}^\top)}$
where $\Theta_i$ and $R_i$ are related via the Choi-Jamio\l{}kowski isomorphism. The assumption that the utilities are in $[-1,1]$ implies that, at each time $t$, the eigenvalues of player $i$'s gain matrix
$
\Theta_i ({\rho^t_{-i}} ^\top)
$ are in $[-1,1]$ for each player $i$.
Then, if player $i$ were to run MMWU with fixed stepsize $\eta \leq 1$ for $T$ timesteps, she would accumulate average regret 
$\leq \eta + \frac{\ln n}{\eta T}$ (see, e.g. \cite{arora2007combinatorial}).
Thus, after $T = \frac{4 \ln n}{\epsilon^2}$ timesteps of running MMWU with fixed stepsize $\eta = \frac{\epsilon}{2}$, she is guaranteed to have average regret $\leq \epsilon$.
    Finally, by Theorem \ref{thm:_sep-QCCE=timeave}\ref{thm:_QCCEa}, the time-averaged joint history of play is a separable $\epsilon$-\ref{QCCE}.

\section{No-Regret Learning in Two-Player Quantum Zero-Sum Games}
Next, we consider the special case of quantum zero-sum games, where the utility is defined such that the sum of all players' payoffs is always zero. In this section, we restrict ourselves to the two-player case in order to present an analogue of a standard classical result - that no-external-regret algorithms go to the set of Nash equilibria in two-player zero-sum games. The main result in this section shows that no-external-regret algorithms converge to the set of \ref{QNE} in two-player quantum zero-sum games.

In a two-player quantum zero-sum game, Alice and Bob play density matrices $\rho\in \mathrm{D}(\A)$ and $\sigma\in \mathrm{D}(\B)$ respectively. 
For notational simplicity, we depart from previous convention and say that Alice's payoff is $u_A(\rho, \sigma) = \innerprod{\rho}{\Theta(\sigma)} = \Tr(R(\rho \otimes \sigma^\top))$, where $\Theta:\mathrm{L}(\B)\rightarrow \mathrm{L}(\A)$ and $R \in \mathrm{L}(\A\otimes\B)$ are related by the Choi-Jamio\l{}kowski isomorphism. By the definition of zero-sum, 
Bob receives payoff $u_B(\rho, \sigma) = -\innerprod{\rho}{\Theta(\sigma)}$. We begin by showing that \ref{QNE}s attain the value of the game in a two-player quantum zero-sum~game.
\begin{theorem}[Quantum Minimax Theorem]\label{thm:quantumminimax}
    Every two-player quantum zero-sum game has a well-defined value, i.e., all \ref{QNE}s attain the same utility
    \begin{equation}
    \label{eqn:maximin}
         v = \max_{\rho\in \da} \min_{\sigma\in \db} \langle \rho, \Theta(\sigma) \rangle =  \min_{\sigma\in \db} \max_{\rho\in \da} \langle \Theta^\dagger(\rho), \sigma \rangle.
    \end{equation}
    Moreover, the set of \ref{QNE}s is the product of two spectrahedra, i.e.,
    \begin{equation}
    \label{eqn:_QNEchar}
        \text{\ref{QNE}s}
        =
        \{\rho \in \da: \Theta^\dagger(\rho) \succeq v I_\B\} \times \{\sigma \in \db: \Theta(\sigma) \preceq v I_\A \}.
    \end{equation}
    
\end{theorem}
\begin{proof}
    The equivalence of max-min and min-max in Equation \eqref{eqn:maximin} comes as a direct consequence of Von Neumann's minimax theorem which holds for compact convex sets \cite{von1928theorie,accardi2020neumann}. Equation \eqref{eqn:_QNEchar} was proven in \cite{ickstadt2022semidefinite}, but for completeness we provide a simplified proof of it here that also proves along the way that all \ref{QNE}s attain the max-min value.
    First, introducing an auxiliary variable $t$, the max-min term in \eqref{eqn:maximin} can be rewritten as 
    \begin{equation}
\begin{split}
     \max_{\rho, t} \qquad & t \\
    \text{s.t.} \qquad &  \langle \rho, \Theta(\sigma) \rangle \geq t \quad \forall \sigma \in \db \\
    \qquad & \rho \in \da,
\end{split}
\end{equation}
which can in turn be rewritten as



\begin{equation}
\label{primal}
     \max_{\rho, t} \left\{t :
    \Theta^{\dagger}(\rho) \succeq t I_\B, \;
    \rho \in \da
    \right\}.
\end{equation}

The dual of this semidefinite program is given by


\begin{equation}
\label{dual}
     \min_{\sigma, t'} \left\{ t'
    : \Theta(\sigma) \preceq t' I_\A, \;
    \sigma \in \db\right\}.
\end{equation}

For proving one direction of Equation \eqref{eqn:_QNEchar}, suppose that $(\rho^*, \sigma^*)$ is a \ref{QNE}, i.e., $$\lambda_{\max}(\Theta(\sigma^*)) = \innerprod{\rho^*}{\Theta(\sigma^*)} = \lambda_{\min}(\Theta^\dagger(\rho^*)).$$ This implies that $\Theta(\sigma^*) \preceq \innerprod{\rho^*}{\Theta(\sigma^*)} I_\A$ and $\innerprod{\rho^*}{\Theta(\sigma^*)}I_\B \preceq \Theta^\dagger(\rho^*)$, which respectively imply that $(\sigma^*, \innerprod{\rho^*}{\Theta(\sigma^*)})$ is feasible for \eqref{dual} and $(\rho^*, \innerprod{\rho^*}{\Theta(\sigma^*)})$ is feasible for \eqref{primal}. But since these programs are a primal-dual pair and the primal-dual feasible solutions attain the same objective value, we have that the dual-feasible solution $(\sigma^*, \innerprod{\rho^*}{\Theta(\sigma^*)})$ and the primal-feasible solution $(\rho^*, \innerprod{\rho^*}{\Theta(\sigma^*)})$ are in fact optimal for the dual \eqref{dual} and the primal \eqref{primal} respectively. This means that the utility $\innerprod{\rho^*}{\Theta(\sigma^*)}$ is equal to the max-min value $v$, and that $(\rho^*, \sigma^*)$ satisfies $\Theta(\sigma^*) \preceq vI_\A$, $\Theta^\dagger(\rho^*) \succeq vI_\B$. 

To prove the other set inclusion in Equation \eqref{eqn:_QNEchar}, suppose that $(\rho^*, \sigma^*)$ satisfies  
\[
    \Theta^\dagger(\rho^*) \succeq vI_\B, 
    \qquad 
    \Theta(\sigma^*) \preceq vI_\A.
\]
Taking inner product of the first inequality with $\sigma^*$ and the second inequality with $\rho^*$ gives us that $\innerprod{\rho^*}{\Theta(\sigma^*)} \geq v$ and $\innerprod{\rho^*}{\Theta(\sigma^*)} \leq v$ respectively, which together imply that $\innerprod{\rho^*}{\Theta(\sigma^*)} = v$. Substituting this fact back into the two inequalities gives
\[
    \lambda_{\min}(\Theta^\dagger(\rho^*)) \geq \innerprod{\rho^*}{\Theta(\sigma^*)},
    \qquad
    \lambda_{\max}(\Theta(\sigma^*)) \leq \innerprod{\rho^*}{\Theta(\sigma^*)},
\]
i.e., that $(\rho^*, \sigma^*)$ is a \ref{QNE}.
\end{proof}


We can use Theorem \ref{thm:quantumminimax} to show the main result of this section, that no-external-regret dynamics converge to \ref{QNE} in two-player quantum zero-sum games.

\begin{theorem}\label{thm:noregretnash}
    For any two-player quantum zero-sum game, the limit points of the product of time-averaged individual histories of play $\left\{\frac{1}{T}\sum_{t=1}^T \rho_t, \frac{1}{T}\sum_{t=1}^T \sigma_t \right\}_T$ generated by no-external-regret algorithms are separable \ref{QNE}. In particular, if all players update their strategies with a no-external-regret algorithm that guarantees a time-averaged regret of $\leq \epsilon$ after $T$ timesteps, then the time-averaged individual sequences of play after $T$ timesteps is a separable $2\epsilon$-\ref{QNE}.

\end{theorem}
\begin{proof}
    For any no-external-regret algorithm, we can select parameters such that, at time $T$, each player's time-averaged regret is at most $\epsilon$. Consider the regret for the sequences of play of Alice (denoted by $\rho$) and Bob (denoted by $\sigma$) respectively. Them we have that:
    \begin{equation}
    \begin{split}\label{eqn:maximin_inequality}
        \min_{\sigma} \frac{1}{T}\sum_{t=1}^T \langle \rho_t, \Theta(\sigma) \rangle +\epsilon 
        \geq \frac{1}{T}\sum_{t=1}^T \langle \rho_t, \Theta(\sigma_t) \rangle \geq \max_{\rho} \frac{1}{T}  \sum_{t=1}^T \langle \rho, \Theta(\sigma_t) \rangle - \epsilon.
    \end{split}
    \end{equation}
    Next, let $\overline{\rho} = \frac{1}{T}\sum_{t=1}^T \rho_t$ and $\overline{\sigma} = \frac{1}{T}\sum_{t=1}^T \sigma_t$
    so that Equation \ref{eqn:maximin_inequality} can be written as
    \begin{align}
        \min_{\sigma} \langle \overline{\rho}, \Theta(\sigma) \rangle +\epsilon &\geq \max_{\rho} \langle \Theta^\dagger(\rho), \overline{\sigma} \rangle - \epsilon \label{eqn:inequalityadjoint}
    \end{align}
    By taking the maximum over $\rho$ for Equation \ref{eqn:inequalityadjoint}, we obtain the following:
    \begin{align*}
        \max_{\rho} \min_{\sigma} \langle \rho, \Theta(\sigma) \rangle &\geq \min_{\sigma} \langle \overline{\rho}, \Theta(\sigma) \rangle\\
        &\geq \max_{\rho} \langle \Theta^\dagger(\rho), \overline{\sigma} \rangle - 2\epsilon\\
         &\geq \min_{\sigma} \max_{\rho} \langle \Theta^\dagger(\rho), \sigma \rangle - 2\epsilon
    \end{align*}
    
    

    By Theorem \ref{thm:quantumminimax}, the left-hand side of the inequality is $v$, the value of the game. Now let us consider Nash strategies, which are strategies for each player that achieve the minimax value regardless of the other player's strategy. Thus, since the time-average value of $\sigma$, $\overline{\sigma}$, satisfies the maximin inequalities above up to a $2\epsilon$ error, it is a $2\epsilon$-approximate Nash strategy for Bob by Theorem \ref{thm:quantumminimax}. A similar argument holds for the case of Alice with $\overline{\rho}$, and thus 
we have that the time average values $(\overline{\rho}, \overline{\sigma})$ are a $2\epsilon$-\ref{QNE} strategy of the zero-sum game, and $\langle\overline{\rho}, \Theta(\overline{\sigma})\rangle$ is the $2\epsilon$-equilibrium value of the game.  

Finally, Theorem \ref{thm:_sep-QCCE=timeave}\ref{thm:_QCCEb}, we have that for any separable \ref{QCCE} there exists no-external-regret algorithms that converge to that \ref{QCCE}. Thereafter, we can take the marginals over the players' joint history of play to obtain a \ref{QNE} of the game.
\end{proof}

We can, with similar reasoning to Remark \ref{remark:generalQG}, give an explicit convergence rate for players using MMWU to update their strategies:

\begin{remark}
    For any two-player quantum zero-sum game, if utilities are in $[-1,1]$ and all players use MMWU with fixed stepsize $\eta = \frac{\epsilon}{4}$ to update their strategies, for any $\epsilon \leq 4$, the product of their time-averaged individual sequences of play $\left(\frac{1}{T}\sum_{t=1}^T \rho_t, \frac{1}{T}\sum_{t=1}^T \sigma_t \right)$ after $T = \frac{16 \ln n}{\epsilon^2}$ steps is an $\epsilon$-\ref{QNE}.
\end{remark}

\section{No-Regret Learning in Polymatrix Quantum Zero-Sum Games}
A key question in the quantum game regime is whether there exist classes of multiplayer games where quantum Nash equilibria are tractable and can be converged to via no-regret learning. As it turns out, in classical polymatrix zero-sum games, \cite{cai2011minmax, daskalakis2009network} show that no-external-regret learning converges in time-average to Nash equilibria. In this section we show an analogous result: in polymatrix quantum zero-sum games, no-external-regret learning converges in time-average to approximate \ref{QNE}.
In order to show this, we first need to define the notion of a polymatrix quantum zero-sum game. 
    \begin{definition}\label{def:polymatrix}
        A polymatrix quantum game $\mathcal{G}$ is a game defined on an undirected graph $(V, E)$ such that the following holds:
        \begin{itemize}
            \item The vertices (or nodes) $V = \{1,\dots,k\}$ represent players, and edges $E$ represent two-player quantum games \eqref{QG:_multilinear} between a pair of players $(i,j)$, where $i\neq j$.
            \item Each player $i\in V$ has register $\mathcal{H}_i$.
            \item For each edge $(i,j)\in E$, we associate a two-player quantum game \eqref{QG:_multilinear} between players $i$ and $j$ where player $i$ has register $\H_i$ and utility tensor $R_{ij}$, while player $j$ has register $\H_j$ and utility tensor $R_{ji}$. Where $\rho_{ij} := \Tr_{-ij} \rho = \Tr_{i}(\Tr_{-j} \rho)) \in \mathrm{D}(\H_i \otimes \H_j)$ is the joint state of the two players' registers, player $i$'s utility from this two-player game is then $u_{ij}(\rho_{ij}) = \Tr(\rho_{ij}R_{ij})$, while player $j$'s utility from this two-player game is $u_{ji}(\rho_{ij}) = \Tr(\rho_{ij}R_{ji})$.
            \item For each joint state $\rho \in {\mathrm{D}}(\bigotimes_i \H_i)$, the total utility attained by player $i \in V$ under $\rho$ is $u_i(\rho) = \sum_{(i, j) \in E} u_{ij}(\rho_{ij}) = \sum_{(i,j) \in E} \Tr(\rho_{ij} R_{ij})$.
            \item Finally, the game $\mathcal{G}$ is called zero-sum if for all joint states $\rho \in {\mathrm{D}}(\bigotimes_i \H_i)$, we have that $\sum_{i\in V} u_i(\rho) = 0$.
            \end{itemize}
    \end{definition}
Note that this definition refers to the zero-sum property of the game in a global sense, as opposed to the stronger notion of \emph{pairwise} zero-sum polymatrix quantum games, the definition of which includes the additional constraint that every two-player edge game is a quantum zero-sum game. 


We next establish a lemma stating that a polymatrix quantum game is also a quantum game in the sense of our earlier definition.

\begin{lemma}\label{lem:polymatrixQG}
    A polymatrix quantum game is also a quantum game in the sense of \eqref{QG:_multilinear}.
\end{lemma}

\begin{proof}
The utility of player $i$ can be expressed as 
    \begin{equation*}
        u_i(\rho) =
        \sum_{j: (i,j) \in E}
        u_{ij}(\rho_{ij}) =
        \sum_{j: (i,j) \in E} 
        \Tr(\rho_{ij} R_{ij}). 
    \end{equation*}
    Subsequently, 
    \begin{equation*}
    \begin{split}
        \sum_{j: (i,j) \in E} 
        \Tr(\rho_{ij} R_{ij}) 
        = \sum_{j: (i,j) \in E} 
        \Tr((\Tr_{-ij}\rho) R_{ij}) =
        \sum_{j: (i,j) \in E}
        \innerprod{\Tr_{-ij} \rho}{R_{ij}},
        \end{split}
    \end{equation*}
    and finally
        \begin{equation*}
    \begin{split}
        \sum_{j: (i,j) \in E}
        \innerprod{\Tr_{-ij} \rho}{R_{ij}} =
        \sum_{j: (i,j) \in E}
        \innerprod{\rho}{R_{ij} \otimes I_{-ij}} =
        \Tr(
            \left(
            \sum_{j: (i,j) \in E}
                R_{ij} \otimes I_{-ij}
            \right)
            \rho
        ),
    \end{split}
    \end{equation*}
    so setting $R_i := \sum_{j:(i,j)\in E} R_{ij} \otimes I_{-ij}$ gives $u_i(\rho) = \Tr(R_i \rho) \; \forall i$, which fits into the \ref{QG:_multilinear} formulation.
\end{proof}

Next, we prove a property connecting \ref{QCCE}s and \ref{QNE}s in the class of polymatrix quantum zero-sum games. 

\begin{lemma}\label{lem:polymatrixmarginal}
    Let $\mathcal{G}$ be a polymatrix quantum zero-sum game. For any joint state $\rho$ that is a \ref{QCCE}, its \emph{marginalized state} $\Hat{\rho}$ defined by
    \begin{equation*}
        \hat{\rho} = \bigotimes_{i \in [n]} \hat{\rho}_i, \qquad
        \hat{\rho}_i = \Tr_{-i} \rho
    \end{equation*}
    is a \ref{QNE} of $\mathcal{G}$. 
\end{lemma}

\begin{proof}
    First note that $\forall i \in [k]$, $\forall \rho_i' \in \mathrm{D}(\H_i)$ we have that
    \begin{equation}
    \label{polymatrix_key}
        u_i(\rho_i' \otimes \hat{\rho}_{-i})
        =
        u_i(\rho_i' \otimes \Tr_i \rho).
    \end{equation}
    This is because if the joint state on all registers is $\rho_i' \otimes \hat{\rho}_{-i}$, then for each $j : (i,j) \in E$ the joint state on player $i$ and $j$'s register is $
    \Tr_{-ij}
    (\rho_i' \otimes \hat{\rho}_{-i}) 
    = \rho_i' \otimes \Tr_{-j} (\hat{\rho}_{-i})
    =  \rho'_i \otimes \hat{\rho}_j,
    $ and thus on the left-hand side of the equation we have that player $i$'s expected utility given this joint state is $u_i(\rho_i' \otimes \hat{\rho}_{-i}) = \sum_{j : (i,j) \in E} u_{ij}(\rho_i' \otimes \hat{\rho}_j)$. 
    On the other hand, if the joint state on all registers is $\rho_i' \otimes \Tr_i \rho$, then for each $j : (i,j) \in E$ we have that the joint state on player $i$ and $j$'s register is also $
    \Tr_{-ij}
    (\rho_i' \otimes \Tr_i \rho) 
    = \rho_i' \otimes \Tr_{-j} (\Tr_i \rho)
    = \rho_i' \otimes \Tr_{-j} \rho
    = \rho_i' \otimes \hat{\rho}_j,
    $ so on the right hand side of the equation we have that player $i$'s expected utility given this joint state is also $u_i(\rho_i' \otimes \Tr_i \rho) = \sum_{j : (i,j) \in E} u_{ij}(\rho_i' \otimes \hat{\rho}_j)$.

    Now fix an $i \in [k]$ and a $\rho_i' \in \mathrm{D}(\H_i)$. By \eqref{polymatrix_key} and the fact that $\rho$ is a \ref{QCCE}, we have that $u_i(\rho) \geq u_i(\rho_i' \otimes \Tr_i \rho) = u_i(\rho_i' \otimes \hat{\rho}_{-i})$, and also that $u_j(\rho) \geq u_j( \hat{\rho}_j \otimes \Tr_j \rho) = u_j(\hat{\rho}) \; \forall j$. Then, summing up the utilities attained by each player on the joint state $\rho$ and using the fact the $\mathcal{G}$ is zero-sum, we have that
    \begin{align*}
        0 
        =
        \sum_{j \in [n]} u_j(\rho)
        &= \sum_{j \neq i} u_j(\rho) + u_i(\rho)  \geq \sum_{j \neq i} u_j(\hat{\rho}) + u_i (\rho_i' \otimes \hat{\rho}_{-i}) = -u_i (\hat{\rho}) + u_i(\rho_i' \otimes \hat{\rho}_{-i}),
    \end{align*}
    i.e., that $u_i(\hat{\rho}) \geq u_i (\rho_i' \otimes \hat{\rho}_{-i}).$ Since this holds for any given $i \in [n]$ and $\rho_i' \in \mathrm{D}(\H_i)$, $\hat{\rho}$ is a \ref{QNE}.
\end{proof}

Finally, we use the previously proved result about convergence to \ref{QCCE}s in general quantum games, in conjunction with the lemmas presented above, to prove convergence of no-external-regret algorithms to \ref{QNE}s in polymatrix quantum zero-sum games. This generalizes the analogous result of \cite{cai2011minmax} from classical polymatrix zero-sum games to the quantum setting.

\begin{theorem}
   If all players in a polymatrix quantum zero-sum game (Definition \ref{def:polymatrix}) use no-external-regret algorithms, then the product of their time-averaged individual histories of play converges to the set of \ref{QNE}. In particular, if all players update their strategies with a no-external-regret algorithm that guarantees time-averaged regret of $\leq \epsilon$ after $T$ timesteps, then the time-averaged joint history of play after $T$ timesteps is a separable $k\epsilon$-\ref{QNE}.
\end{theorem}
\begin{proof}
    Suppose that after $T$ iterations of running no-external-regret algorithms, every player has time-average regret $\leq \epsilon = \epsilon(T)$. Let $\rho^t = \bigotimes_{i=1}^k \rho_i^t$ denote the strategy profile (product distribution) at time $t$, and let $\overline{\rho} = \overline{\rho}(T) := \frac{1}{T} \sum_{t=1}^T \rho^t$ be the time-averaged history of these strategy profiles. Moreover, note that if players $i$ and $j$ play with strategies $\rho_i$ and $\rho_j$ respectively, we can write the utility for player $i$ in the form $u_{ij}(\rho_i \otimes \rho_j)$. For any $\rho_i'\in \mathrm{D}(\H_i)$ we can write:

    \begin{align*}
        \frac{1}{T} \sum_{t=1}^T \sum_{(i,j)\in E} u_{ij}(\rho_i' \otimes \rho_j^t) = \sum_{(i,j)\in E} u_{ij}(\rho_i' \otimes \overline{\rho}_j)
    \end{align*}
    Let $z_i$ be the best response of $i$ if all other players use $\overline{\rho}_j$. Then for all $i$ and any $\rho_i'\in \mathrm{D}(\H_i)$,

    \begin{align*}
        \sum_{(i,j)\in E} u_{ij}(z_i \otimes \overline{\rho}_j) \geq \sum_{(i,j)\in E} u_{ij}(\rho_i' \otimes \overline{\rho}_j).
    \end{align*}
Next, by the no-external-regret property, we have that
\begin{align*}
     \frac{1}{T} \sum_{t=1}^T \sum_{(i,j)\in E} u_{ij}(\rho_i^t \otimes \rho_j^t) \geq \frac{1}{T} \sum_{t=1}^T \left(\sum_{(i,j)\in E} u_{ij}(z_i \otimes \rho_j^t) \right) - \epsilon  =  \sum_{(i,j)\in E} u_{ij}(z_i \otimes \overline{\rho}_j) - \epsilon
\end{align*}
Summing both sides of the above over all $i\in V$, we have from the LHS that
\begin{align*}
     \sum_{i\in V}\left(\frac{1}{T} \sum_{t=1}^T \sum_{(i,j)\in E} u_{ij}(\rho_i^t \otimes \rho_j^t) \right) = \frac{1}{T} \sum_{t=1}^T  \left(\sum_{i\in V} \sum_{(i,j)\in E} u_{ij}(\rho_i^t \otimes \rho_j^t) \right) = 0,
\end{align*}
which is due to the global zero-sum property of the quantum polymatrix game. Moreover, the sum on the RHS is given by
\begin{align*}
   \sum_{i\in V} \sum_{(i,j)\in E} u_{ij}(z_i \otimes \overline{\rho}_j) - k\epsilon
\end{align*}
since there are $k$ players. Combining the two, and using the fact that the LHS is at least as large as the RHS,

\begin{align*}
    0 \geq \sum_{i\in V} \sum_{(i,j)\in E} u_{ij}(z_i \otimes \overline{\rho}_j) - k\epsilon \implies k\epsilon \geq \sum_{i\in V} \sum_{(i,j)\in E} u_{ij}(z_i \otimes \overline{\rho}_j).
\end{align*}

We now show that each player $i$ playing $\overline{\rho}_i$ is a $k\epsilon$- approximate \ref{QNE}. Note that if each player $i$ plays $\overline{\rho}_i$, the sum of all players' payoffs is $0$, i.e. 
\[
\sum_{i\in V} \sum_{(i,j)\in E} u_{ij}(\overline{\rho}_i \otimes \overline{\rho}_j) = 0.
\] 
Hence we have that 
\[
 k\epsilon \geq \sum_{i\in V} \left( \sum_{(i,j)\in E} u_{ij}(z_i \otimes \overline{\rho}_j) - \sum_{(i,j)\in E} u_{ij}(\overline{\rho}_i \otimes \overline{\rho}_j)\right)
\] 
However, the sum is over non-negative numbers since the $z_i$s are best responses. We have a sum of non-negative numbers bounded by $k\epsilon$, so for any $i\in V$,
\[
k\epsilon \geq  \sum_{(i,j)\in E} u_{ij}(z_i \otimes \overline{\rho}_j) - \sum_{(i,j)\in E} u_{ij}(\overline{\rho}_i \otimes \overline{\rho}_j) \geq 0.
\]
Thus, for all $i$, if all other players $j$ play $\overline{\rho}_j$, the payoff given by playing the best response is at most $k\epsilon$ better than the payoff obtained by playing $\overline{\rho}_i$. Hence it is a $k\epsilon$-\ref{QNE} for each player $j$ to play $\overline{\rho}_j$.
\end{proof}
This result gives us a decentralized way of arriving at quantum Nash equilibria in a broader class of multi-player games, i.e., that of polymatrix quantum zero-sum games. Exploring if there are other classes of multi-player games for which \ref{QNE} are tractable is left to future work.

\begin{remark}
    Consider a $k$-player polymatrix quantum zero-sum game with utilities in $[-1,1]$ and let $n$ be the largest dimension of the players' registers.  For any $\epsilon \leq 2k$, if each player uses MMWU with fixed stepsize $\eta = \frac{\epsilon}{2k}$, the product of their time-averaged individual sequences of play $\left(\frac{1}{T}\sum_{t=1}^T \rho_1^t,\dots, \frac{1}{T}\sum_{t=1}^T \rho_k^t \right)$ after $T = \frac{4k^2 \ln n}{\epsilon^2}$ steps is an $\epsilon$-\ref{QNE}. 
\end{remark}
The reasoning for the above convergence rate is similar to Remark \ref{remark:generalQG}. However, since an algorithm that achieves $\epsilon$-regret gives a $k\epsilon$-\ref{QNE}, we require running the algorithm until $\frac{\epsilon}{k}$ regret is achieved instead.
\section{MMWU Experiments}
\label{sec:experiments}

In this section, we consider learning using the specific no-external-regret algorithm, MMWU (Algorithm \ref{alg:mmwu}), and present several experiments that corroborate our theoretical results about time-averaged convergence to equilibria. For two-player zero-sum quantum games, we also present some plots showcasing the day-to-day behavior of the iterates.



First, in Figure \ref{fig:exploit} we show the exploitability (as defined in Section \ref{subsec:qeq}) of MMWU in both general and zero-sum quantum games.  
For the case of general games, we consider the maximum individual exploitability of the time-averaged joint strategy for both players, which we term the ``\ref{QCCE}-exploitability'' of the players' strategies, while in the case of zero-sum games we consider the maximum individual exploitability of the product of the time-averaged individual sequences of play, which we term the ``\ref{QNE}-exploitability''. We are concerned with the maximum over the individual exploitabilities of each of the players since if each player attains $\epsilon$-exploitability, then all players are at an $\epsilon$-\ref{QCCE}/\ref{QNE}. In both plots, we use the doubling trick to run MMWU. The maximum individual exploitabilities go to zero or remain close to zero, implying time-averaged convergence to \ref{QCCE} and \ref{QNE} respectively.

\begin{figure}[!htb]
    \centering
    \begin{minipage}{.49\linewidth}
      \centering
      \includegraphics[width=.95\linewidth]{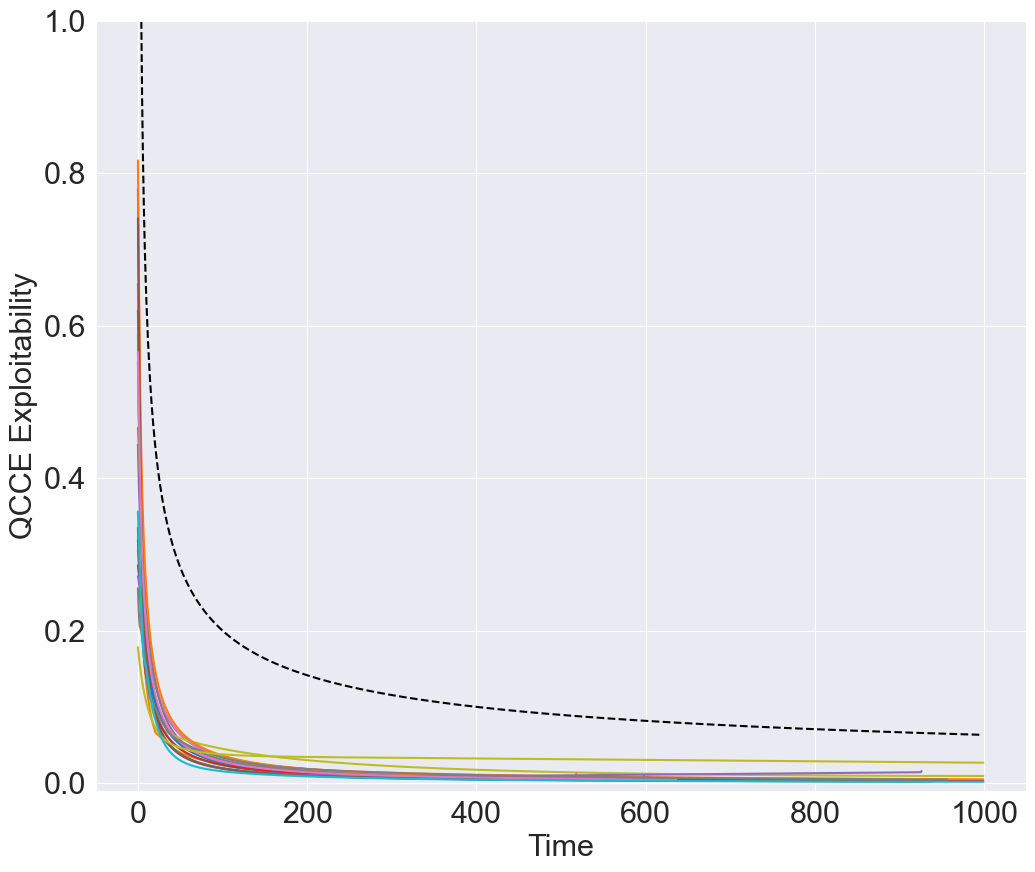}
     General quantum games
    \end{minipage}
    \begin{minipage}{.49\linewidth}
      \centering
      \includegraphics[width=.95\linewidth]{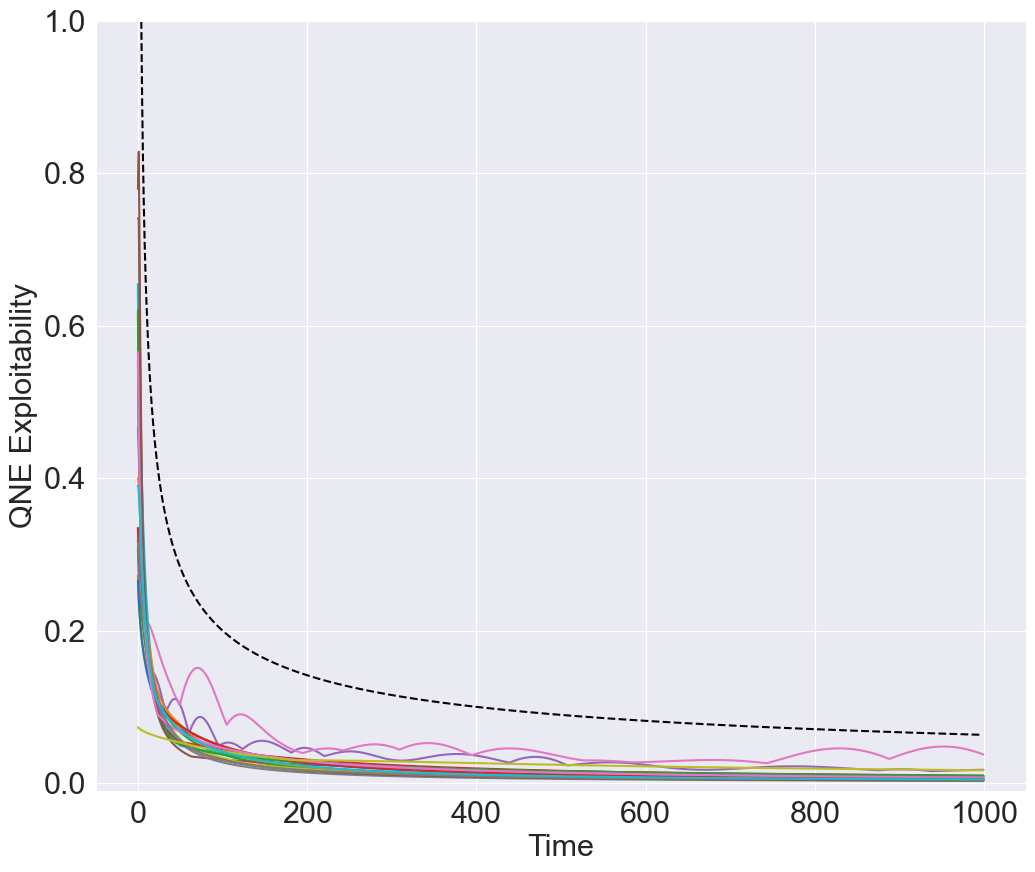}
      Zero-sum quantum games
    \end{minipage}
    \caption{Maximum individual exploitability of time-averaged strategies of players using MMWU in 20 randomly generated $\mathbb{C}^2\otimes\mathbb{C}^2$ quantum games. The black dotted line denotes the theoretical upper-bound on the exploitability.}
    \label{fig:exploit}
\end{figure}

Next, we present some indicative examples that elucidate the behavior of MMWU in two-player quantum zero-sum games.
We see that in general, the trajectories of the joint state of the players either oscillate or go to a point on the boundary, and showcase this behavior alongside the time-averaged values of the trajectories in Figures \ref{fig:qzsg_osc} and \ref{fig:qzsg_conv}. In order to represent time on the Bloch sphere, we use a color gradient from green to blue (light green denotes time $t=0$, dark blue denotes time $t=4000$). From the examples, even in the relatively well-studied case of MMWU, it is clear that some interesting types of behavior can be observed in quantum zero-sum games and beyond. \edits{The code used to generate our MMWU experiments can be found in the following Github repository: \href{https://github.com/ryanndelion/No-Regret-Learning-Quantum-Games}{https://github.com/ryanndelion/No-Regret-Learning-Quantum-Games}.}

\begin{figure}[!htb]
    \centering
    \begin{minipage}{.32\linewidth}
      \centering
      \includegraphics[width=.73\linewidth]{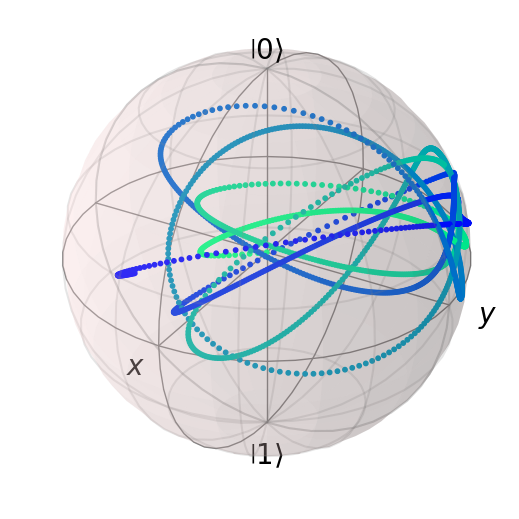}\\
      Trajectory of one player's strategy plotted on the Bloch sphere
    \end{minipage}
    \begin{minipage}{.33\linewidth}
      \centering
      \includegraphics[width=.95\linewidth]{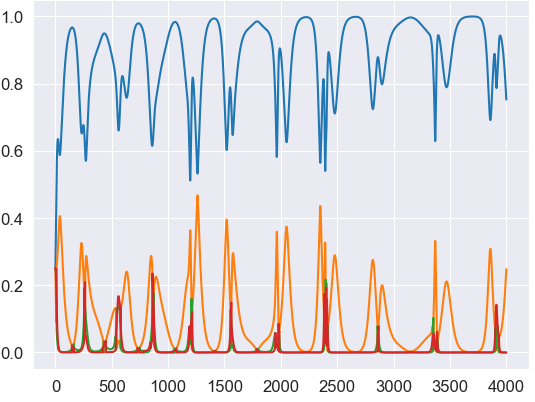}
      Eigenvalues of the players' joint state over time
    \end{minipage}
    \begin{minipage}{.33\linewidth}
      \centering
      \includegraphics[width=.95\linewidth]{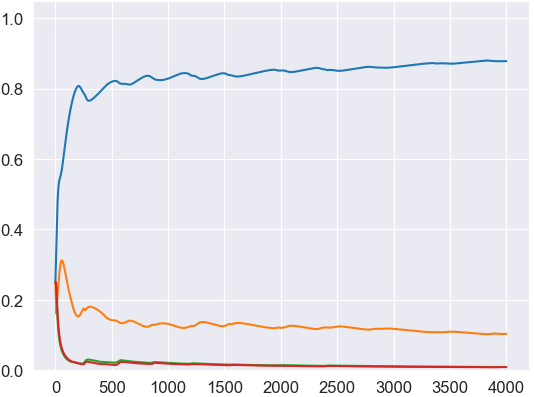}
      Eigenvalues of the time-averaged joint state over time
    \end{minipage}
    \caption{Example of oscillatory behaviour of MMWU in two-player quantum zero-sum games. Time is represented using a gradient from green to blue on the Bloch sphere.}
    \label{fig:qzsg_osc}
\end{figure}

\begin{figure}[!htb]
    \centering
    \begin{minipage}{.32\linewidth}
      \centering
      \includegraphics[width=.73\linewidth]{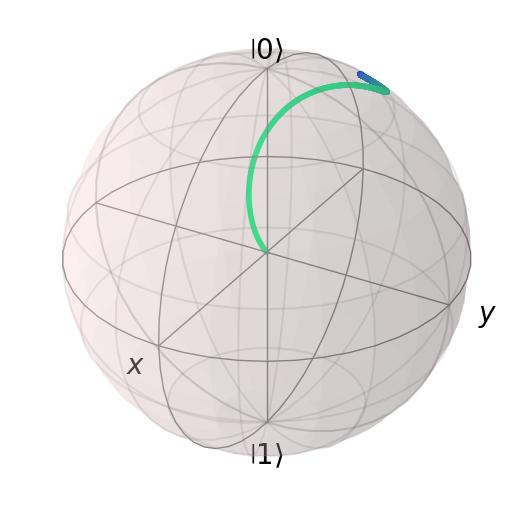}\\
      Trajectory of one player's strategy plotted on the Bloch sphere
    \end{minipage}
    \begin{minipage}{.33\linewidth}
      \centering
      \includegraphics[width=.95\linewidth]{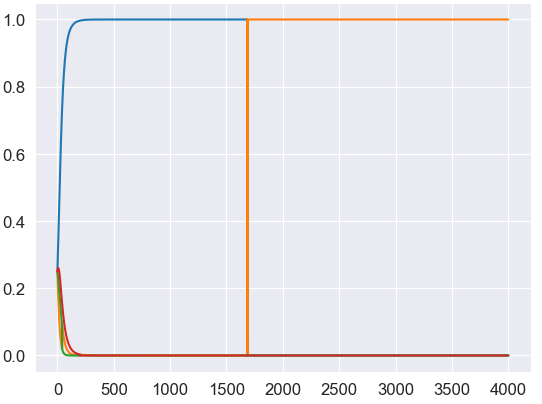}
      Eigenvalues of the players' joint state over time
    \end{minipage}
    \begin{minipage}{.33\linewidth}
      \centering
      \includegraphics[width=.95\linewidth]{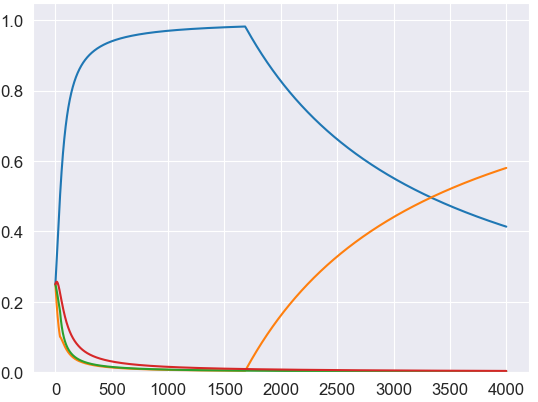}
      Eigenvalues of the time-averaged joint state over time
    \end{minipage}
    \caption{Example of MMWU converging to the boundary (i.e., pure states) in two-player quantum zero-sum games. Time is represented using a gradient from green to blue on the Bloch sphere.}
    \label{fig:qzsg_conv}
\end{figure}


\section{Discussion and Future Work}
In this work we provide a general class of quantum games that fits with and subsumes prior formulations. We explore equilibrium notions in this class of games, inspired by classical solution concepts and $\mathbf{\Phi}$-regret and show an interesting analogy between deviation maps in the classical and quantum settings. We introduce quantum coarse correlated equilibria and show that for general quantum games, the set of separable \ref{QPhiE}s is actually the set of limits points of the time-averaged distribution of joint play when players use no-$\mathbf{\Phi}$-regret algorithms. Moreover, in the two-player and polymatrix zero-sum cases, no-regret algorithms result in convergence to quantum Nash equilibria. Overall, this indicates a rich connection between the worlds of online optimization, classical learning in games, and quantum information theory.

An interesting future direction of our work is to study general $\mathbf{\Phi}$-equilibria in other classes of quantum games, and the capability of modifications to the standard no-regret learning paradigm that can arrive at these equilibria. Specifically, designing implementable algorithms that converge to quantum correlated equilibria remains an important task, given that similar approaches have been successful in the classical regime.
Additionally, the quantum game formulation allows for entangled equilibria not reachable via the standard paradigm of learning in games\edits{, examples of which were constructed in Appendix \ref{sec:maxent}.}
Investigating these entangled equilibria and how they can be computed or learnt by distributed agents is a tantalizing direction for future work. \edits{In the setting of classical games, several approaches have utilized coupled or correlated mechanisms to converge to different (and often better) equilibria than their uncoupled counterparts \cite{balcan2013circumventing,zhang2023steering}. Thus, studying a variant of no-regret learning which utilizes a mediator or correlating mechanism seems to be a reasonable initial approach to learning entangled equilibria in our formulation of quantum games.
}

\section*{Acknowledgments}
 
This research is supported in part by the National Research Foundation, Singapore and the Agency for Science, Technology and Research (A*STAR) under its Quantum Engineering Programme NRF2021-QEP2-02-P05, and by the National Research Foundation, Singapore and DSO National Laboratories under its AI Singapore Program (AISG Award No: AISG2-RP-2020-016), 
grant PIESGP-AI-2020-01, AME Programmatic Fund (Grant No.A20H6b0151) from A*STAR and Provost’s Chair Professorship grant RGEPPV2101. Wayne Lin and Ryann Sim gratefully acknowledge support from the SUTD President's Graduate Fellowship (SUTD-PGF).
\bibliography{doi_quantum_refs}
\bibliographystyle{quantum}

\newpage
\appendix
\onecolumn
\section{Examples of entangled equilibria}\label{sec:maxent}

\edits{
In this section we showcase examples of entangled \ref{QCCE}s which are unapproachable by the decentralized no-regret learning paradigm. 
The idea is to use \emph{maximally-entangled games}, in which payoffs are assigned to states in a maximally-entangled basis instead of the standard product-state basis as one might think to do when attempting to embed a classical game in the quantum game formulation that we use (indeed, that is precisely what is done in \cite{2012QSGT}). 
We shall first define maximally-entangled games before characterizing the maximally-entangled \ref{QCCE}s in any maximally-entangled game.


For simplicity we consider two-player games where each player has access to a qubit ($n_1  = n_2 = 2$). The Bell states 
\begin{equation*}
\label{bell}
	\begin{split}
		\ket{e_{1}} &= \ket{\phi^+} = \frac{1}{\sqrt{2}}( \ket{00} + \ket{11} ), \\
		\ket{e_{2}} &= \ket{\phi^-} = \frac{1}{\sqrt{2}}( \ket{00} - \ket{11} ), \\
		\ket{e_{3}} &= \ket{\psi^+} = \frac{1}{\sqrt{2}}( \ket{01} + \ket{10} ), \\
		\ket{e_{4}} &= \ket{\psi^-} = \frac{1}{\sqrt{2}}( \ket{01} - \ket{10} )
	\end{split}
\end{equation*}
form a maximally entangled basis for the joint space $\mathcal{H}_1 \otimes \mathcal{H}_2$. We can define a \emph{maximally-entangled game} as follows:

\begin{definition}
    A maximally-entangled (max-ent) game is a two-player \ref{QG:_multilinear} in which the game operators are supported on the rank-1 projectors of a maximally entangled basis $\{\ket{e_k}\}$, i.e., the game operators are given by
    \begin{equation}
        \label{eq:maxent}
		R_1 = \sum_{k} a_{k} \ket{e_{k}} \bra{e_{k}},
		\qquad
		R_2 = \sum_{k} b_{k} \ket{e_{k}} \bra{e_{k}}.
    \end{equation}
\end{definition}



The following theorem characterizes the \ref{QCCE}s in a max-ent game that are mixtures of states in the maximally-entangled basis. Crucially, coarse unilateral deviations from mixtures of maximally-entangled states set the other party's reduced state to the maximally mixed state (scaled identity), and the fact that the game operators are supported only on the rank-1 projectors of the maximally-entangled basis makes the utilities achieved by these deviations easy to compute. These two facts make characterizing when such a state is a QCCE easy.

\medskip
\begin{theorem}
Fix a max-ent game on a maximally-entangled basis $\{\ket{e_k}\}$ and game operators given by \eqref{eq:maxent}.
A mixture of states in the maximally-entangled basis, $\rho^* = \sum_{k} \lambda_{k} \ket{e_{k}} \bra{e_{k}}$, is a \ref{QCCE} of the max-ent game if and only if
	\begin{equation}
        \label{maxent_QCCE}
		\sum_{k} a_{k} \lambda_{k} \geq \frac{1}{n_1 n_2} \sum_{k} a_{k}
		\qquad
		\text{and}
		\qquad
		\sum_{k} b_{k} \lambda_{k} \geq \frac{1}{n_1 n_2} \sum_{k} b_{k}.
	\end{equation}
\end{theorem}

\begin{proof}
	For Player 1, the utility achieved from sticking to $\rho^*$ is given by
	\[
		\Tr(R_1 \rho^*) = \sum_{k} a_{k} \lambda_{k},
	\]
	while the utility achieved from deviating to $\rho_1'$ is given by
	\begin{equation*}
	\begin{split}
		\Tr(R_1 (\rho_1' \otimes \Tr_\A \rho^*))
		&=
		\frac{1}{n_2} \Tr(R (\rho_1' \otimes I_\mathcal{B})) \\
		&=
		\frac{1}{n_2} \Tr(\Tr_{\mathcal{B}}(R) \rho_1') \\
		&=
		\frac{1}{n_1 n_2} \sum_{k} a_{k} \Tr(\rho_1') \\
		&=
		\frac{1}{n_1 n_2} \sum_{k} a_{k},
	\end{split}
	\end{equation*}
	where the first equality is due to the fact that 
 \[\Tr_\mathcal{A} \rho^* = \sum_k \lambda_k \Tr_\mathcal{A}(\ketbra{e_k}) = \sum_k \lambda_k ( \frac{1}{n_2}I_\mathcal{B}) = \frac{1}{n_2}I_\mathcal{B}\] and the third equality is due to the fact that \[\Tr_\mathcal{B}(R) = \sum_{k} a_{k} \Tr_\mathcal{B}(\ket{e_{k}} \bra{e_{k}}) = \frac{1}{n_1} \sum_{k} a_{k} I_\A.\]
	Thus Player 1 has no incentive to do a coarse deviation if and only if
	\[
		\sum_{k} a_{k} \lambda_{k} \geq \frac{1}{n_1 n_2} \sum_{k} a_{k}.
	\]
	We can similarly get the analogous statement for Player 2, and thus $\rho^*$ is a \ref{QCCE} if and only if both the conditions in \eqref{maxent_QCCE} hold.
\end{proof}

As a corollary, we are able to construct pure, entangled \ref{QCCE}s in any common-interest max-ent game.

\begin{corollary}
    Fix a maximally-entangled basis $\{\ket{e_k}\}$ for the joint strategy space and suppose that on playing joint strategy $\rho$ both players get common utility $\Tr(R\rho)$ where
    \[
        R = \sum_k a_k \ketbra{e_k}.
    \]
    Define $k^* := \argmax_k a_k$. Then the pure, entangled joint state
    \[
        \rho* = \ketbra{e_{k^*}}
    \]
    is a \ref{QCCE} of this game.
\end{corollary}
}

\section{Additional Quantum Preliminaries}\label{appsec:prelims}
\paragraph{Quantum states.}
\textit{Pure quantum states} correspond to (typically unit-length normalized) vectors in a Hilbert space $\mathcal{H}$.
The simplest case is that of a qubit, which can be represented by a linear superposition of its two orthonormal basis states. 
These vectors are usually denoted as
$| 0 \rangle = \bigl[\begin{smallmatrix}
1\\
0
\end{smallmatrix}\bigr]$
and
$| 1 \rangle = \bigl[\begin{smallmatrix}
0\\
1
\end{smallmatrix}\bigr]$ in the conventional bra–ket notation introduced by Paul Dirac and together span the qubit's two-dimensional Hilbert space. 
A single qubit $\psi$ can be described by a linear combination of $|0 \rangle$ and $|1 \rangle$
: $| \psi \rangle = \alpha |0 \rangle + \beta |1 \rangle$ 
where $\alpha $ and $\beta$ are the probability amplitudes, i.e., complex numbers such that  $| \alpha |^2 + | \beta |^2 = 1$.

\paragraph{Quantum measurements.} We utilize the generalized measurement formalism known as the positive-operator-valued measure (POVM) in the main text, but for completeness we also present several other key formalisms for quantum measurements.

\textit{Idealized von Neumann measurements.} The approach codified by John von Neumann represents a measurement upon a physical system by a self-adjoint operator on that Hilbert space termed an "observable".
We start by
representing each observable by a Hermitian operator $A$. This operator will have
a complete set of (normalized) eigenvectors $|\lambda_n \rangle $  and associated eigenvalues $\lambda_n$\footnote{That is we have $A|\lambda_n \rangle= \lambda_n  |\lambda_n \rangle$.},
thus we can write $A$ in the form $A = \sum_n\lambda_n  |\lambda_n \rangle \langle \lambda_n |$.
Let's assume, for the moment and for simplicity, that all the eigenvalues are distinct.
The von Neumann description then states that if we perform a measurement of $A$ then
we will find the result of the measurements to be one of the eigenvalues and the
probability for finding any one of these is $P(\lambda_n) = |\langle \lambda_n | \psi \rangle|^2$. 
Whereas in the previous paragraph we chose the observable with $A = 0 |0\rangle \langle 0 | + 1 |1\rangle \langle 1 |$ now we can choose another observable  
$B = \lambda_1 | \lambda_1\rangle \langle \lambda_1 | + \lambda_2 |\lambda_2 \rangle \langle \lambda_2 |$. With a bit of algebra, we can verify that for $n=1, 2:$ $P(\lambda_n) = |\big\langle \lambda_n | ( \alpha |0 \rangle + \beta |1 \rangle )\big\rangle|^2 = |\alpha|^2 \langle |\lambda_n |0 \rangle |^2 + |\beta|^2 \langle |\lambda_n |1 \rangle |^2 + 2 Re \{ \alpha  \beta^* \langle \lambda_n  |0 \rangle  \langle \lambda_n  |1 \rangle^*\}$. In this case the last term, which is known as the interference term, no longer vanishes as in the case of observable $A$. Specifically, the expected utility/measurement will not be in agreement with that of a classical probability   distribution which is in state $| 0 \rangle$ with probability  $|\alpha|^2$ and in state $| 1 \rangle$ with probability  $|\beta|^2$.


\textit{From quantum states to density operators/matrices.} The probability of measuring eigenvalue $\lambda_n$ is 

\[
P(\lambda_n) = |\langle \lambda_n | \psi \rangle|^2 = \langle \lambda_n | \psi \rangle \langle \psi | \lambda_n \rangle = \langle \lambda_n | \rho | \lambda_n \rangle = \Tr ( \rho | \lambda_n \rangle  \langle \lambda_n |) = \Tr (\rho P_n )
\]

where $\rho= |\psi \rangle \langle \psi|$ is the rank-1 projection operator onto the space spanned by the state $|\psi \rangle$ and is called (probability) density of $|\psi \rangle$ and $P_n$ is the projector on the space spanned by the eigenvector\footnote{If the eigenvalues
of $A$ are degenerate, there exists a set of orthonormal eigenvectors, $|\lambda^j_n \rangle$,
which correspond to the same 
$\lambda_n$, then $P(\lambda_n) =  Tr (\rho | \sum_j |\lambda^j_n \rangle  \langle \lambda^j_n |)= Tr (\rho P_n )$, i.e., we use a projector onto the set of states with eigenvalue $| \lambda_n \rangle$.} $|\lambda_n\rangle$. Thus, the overall expected measurement is equal to\footnote{Another set of useful formulas that easily follow in the case of pure states are $P(\lambda_n) = \langle \lambda_n | P_n| \lambda_n \rangle$ and  $\langle A \rangle_\psi = \langle \lambda_n | A| \lambda_n \rangle $.}

\[
\langle A \rangle_\psi = \sum_n \lambda_n P(\lambda_n) = \sum_n \lambda_n \Tr (\rho P_n ) = \Tr (\rho \sum_n \lambda_n P_n ) = \Tr (\rho A).
\]

The set of projectors $P_n$ above have the following three properties: they are Hermitian, they are positive semi-definite operators and they are complete; they sum up to the identity.  These properties have
physical meaning. They represent, respectively, the requirements that
the projectors are observables, that they give non-negative probabilities and that the sum of the probabilities for all possible outcomes must be equal to one. 
 Generalized measurements will correspond to a collection of such projectors without necessarily being orthonormal.

\section{Omitted Proofs}
\medskip 
\begin{lemma}
\label{lem:QCCE_equiv}
    The two definitions of QCCE, namely
    \begin{equation*}
    u_i (\rho) \geq u_i ((\phi_i \otimes \mathbb{I}_{-i})(\rho))
    \tag{dev-QCCE}
\end{equation*}
for all replacement channels $\phi_i: \mathrm{L}(\H_i) \rightarrow \mathrm{L}(\H_i), \; X \mapsto \Tr(X) \rho_i'$ for some $\rho_i' \in \mathrm{L}(\H_i)$ and
\begin{equation*}
    \tag{QCCE}
		u_i (\rho) \geq u_i (\rho_i' \otimes \Tr_i \rho)
	\end{equation*}
	where $\Tr_i : \mathrm{L}(\bigotimes_{i'} \H_{i'}) \rightarrow \mathrm{L}(\bigotimes_{i' \neq i} \H_{i'})$ is the partial trace with respect to player $i$'s subsystem, are equivalent.
\end{lemma}

\begin{proof}
    For any joint state $\rho \in \mathrm{D}(\otimes_i \H_i)$ we can write
	\[	
		\rho = \sum_{k} X_k \otimes Y_k
	\]
	for some $X_k \in \mathrm{L}(\H_i), Y_k \in \mathrm{L}(\H_{-i})$ since $\rho \in \mathrm{L}(\otimes_i \H_i) = \otimes_i \mathrm{L}(\H_i)$. Then when $\phi_i$ is the replacement channel $\phi_i: \mathrm{L}(\H_i) \rightarrow \mathrm{L}(\H_i), \; X \mapsto \Tr(X) \rho_i'$,
	\begin{equation*}
	\begin{split}
		(\phi_i \otimes \mathbb{I}_{-i}) (\rho)
		= \sum_k (\phi_i \otimes \mathbb{I}_{-i})(X_k \otimes Y_k) 
		= \sum_k \phi_i(X_k) \otimes Y_k 
		&= \left(\sum_k \Tr(X_k) \rho_i' \right) \otimes Y_k \\
		&= \rho_i' \otimes \sum_k \Tr(X_k)Y_k \\
		&= \rho_i' \otimes \Tr_i \rho.
	\end{split}
	\end{equation*}
\end{proof}

\medskip
\begin{lemma}
\label{lem:_sup_unifcont_cont} 
	Let $h(x,y): X \times Y \rightarrow \R$ be a continuous function on the product of compact sets $X, Y$. Then $z: X \rightarrow \R$, $z(x) = \sup_{y \in Y} h(x,y)$ is continuous.
\end{lemma}

\begin{proof}
	Since $h$ is a continuous function on a compact domain, it is uniformly continuous. In particular, given any $\epsilon > 0$ $\exists \delta > 0 $ such that $\forall y$, $\forall \norm{x - x'} < \delta$, $\abs{h(x,y) - h(x',y)} < \epsilon$.
	
	Then  $\forall y$, $\forall \norm{x - x'} < \delta$ we have that $h(x,y) \leq h(x',y) + \epsilon \leq z(x') + \epsilon$, which taking supremum over $y \in Y$ on both sides gives us that $z(x) \leq z(x') + \epsilon$ $\forall \norm{x - x'} < \delta$.
	
	On the other hand, we have similarly that $\forall y$, $\forall \norm{x - x'} < \delta$ $h(x'y) \leq h(x,y) + \epsilon \leq z(x) + \epsilon$, so similarly taking supremum over $y$ on both sides gives us that $z(x') \leq z(x) + \epsilon$ $\forall \norm{x - x'} < \delta$.
	
	Combining the last two results gives us that $\abs{z(x') - z(x)} < \epsilon$ $\forall \norm{x - x'} < \delta$. Thus $z$ is continuous.
\end{proof}

\end{document}